\documentclass[submission,copyright,creativecommons]{eptcs}
\providecommand{\event}{AiML 2026} 

\usepackage{aiml26}

\usepackage{iftex}

\ifpdf
  \usepackage{underscore}         % Only needed if you use pdflatex.
  \usepackage[T1]{fontenc}        % Recommended with pdflatex
\else
  \usepackage{breakurl}           % Not needed if you use pdflatex only.
\fi

\usepackage{xcolor}
\usepackage{latexsym}
\usepackage{amsmath}
\usepackage{amssymb}
\usepackage{mathrsfs} % makes \mathscr work properly
\usepackage{tikz}
\usetikzlibrary{calc}

\newcommand{\dis}{\vee}
\newcommand{\vp}{\varphi}

\newcommand{\lb}{\langle}
\newcommand{\rb}{\rangle}

\newcommand{\dfnotion}[1]{\emph{#1}}
\newcommand{\rmbr}[1]{{\rm{(}\/}#1\/{\rm{)}}}

\newcommand{\type}{t}
\newcommand{\sub}{\mathop{\textit{sub}}}
\newcommand{\sur}[2]{#2^{\mathbf{#1}}}

\newcommand{\rsur}[2]{#2\uparrow^{\mathbf{#1}}}
\newcommand{\adp}{\textit{adp}}
\newcommand{\md}{\textit{md}}
\newcommand{\dist}{\textit{dist}}
\newcommand{\qs}{\pmb{q}}
\newcommand{\QS}{\pmb{Q}}

\newcounter{mthm}

\theoremstyle{plain}
\newtheorem{mtheorem}[mthm]{Theorem}

\newcommand{\Log}{\textup{Log}}
\newcommand{\Logd}{\mathop{\Log_d}}
\newcommand{\Loged}{\mathop{\Log^=_d}}
\newcommand{\Logecd}{\mathop{\Log_{\textit{cd}}^=}}
\newcommand{\Logexd}{\mathop{\Log_{\textit{xd}}^=}}

\newcommand{\SIG}{\Phi}
\newcommand{\FL}{\Phi}
\newcommand{\dom}{\mathop{\mathsf{dom}}}

\title{Fusions of One-Variable First-Order Modal Logics}
\author{Roman Kontchakov
\institute{Birkbeck, University of London, UK}
\email{r.kontchakov@bbk.ac.uk}
\and
Dmitry Shkatov
\institute{University of the Witwatersrand,\\ Johannesburg,
South Africa}
\email{dmitry.shkatov@wits.ac.za}
\and 
Frank Wolter
\institute{University of Liverpool, UK}
\email{wolter@liverpool.ac.uk}
}

\newcommand{\titlerunning}{Fusions of One-Variable First-Order Modal Logics}
\newcommand{\authorrunning}{R. Kontchakov, D. Shkatov \& F. Wolter}

\hypersetup{
  bookmarksnumbered,
  pdftitle    = {\titlerunning},
  pdfauthor   = {\authorrunning},
  pdfsubject  = {Modal Logic},               % Consider adding a more appropriate subject or description
  pdfkeywords = {modal logic, fusions, Kripke completeness, global consequence, decidability, finite model property, Diophantine equations} % Uncomment and enter keywords specific to your paper
}

\begin{document}
\maketitle

\begin{abstract}
We investigate preservation results for the independent fusion of one-variable first-order modal logics. We show that, without equality, Kripke completeness and decidability of global and local consequence relations are preserved, under both expanding and constant domain semantics. By contrast, Kripke completeness and decidability are not preserved for fusions with equality and non-rigid constants (or, equivalently, counting up to one), again for the global and local consequence and under both expanding and constant domain semantics. This result is shown by encoding Diophantine equations. Even without equality, the finite model property is  preserved only in the local case. Finally, we view fusions of one-variable modal logics as fusions of propositional modal logics sharing an \textbf{S5} modality and provide a general sufficient condition for transfer of Kripke completeness and decidability (but not of  finite model property).
\end{abstract}

% !TEX root =  fusions-submission.tex

\section{Introduction}
The study of combined modal logics has been an active area of research for many years~\cite{Gabbay2003Fibred,Kurucz2007Combining,DBLP:journals/jsyml/ZanardoSS01}. Only a few methods of combining modal systems have proved to be robust in the sense that key properties of the component logics are always inherited by the resulting combined logic. A central example of such a robust combination is fusion: a basic method in which the component logics are combined without introducing axioms that mix their modal operators, and, correspondingly, without imposing any semantic constraints that relate the distinct accessibility relations. In fact, properties such as Kripke completeness, decidability, finite model property, and Craig interpolation are all preserved under forming the fusion of two (or more) classical propositional modal logics, for both the local and global consequence relations~\cite{Thomason1980,DBLP:journals/jsyml/KrachtW91,FineSchurz1996,DBLP:conf/aiml/Wolter96,DBLP:journals/jair/BaaderLSW02,DBLP:conf/lpar/GhilardiS03,DBLP:journals/iandc/BaaderGT06,DBLP:conf/calco/DahlqvistP11}. 

The aim of this contribution is to initiate a systematic investigation into the extent to which these preservation results for fusions of propositional modal logics can be lifted to (fragments of)
first-order modal logics and corresponding many-dimensional modal logics. We focus in this paper on one-variable first-order modal logics, either without equality or with equality and non-rigid constants, and on corresponding commutators and products of propositional modal logics. These logics have been studied extensively~\cite{DBLP:journals/igpl/GabbayS98,DBLP:journals/tocl/HampsonK15,DBLP:journals/lmcs/KuruczWZ25,BezhanishviliKhan2026,ShehtmanShkatov2019} as they are often computationally and semantically better behaved than full first-order modal logics. For fusions without equality our main result is as follows:
\begin{mtheorem}[Equality-Free Fusion]\label{thm:1}
	Kripke completeness and decidability are preserved under fusions of Kripke complete one-variable first-order modal logics, for both the local and global consequence and under both expanding and constant domain semantics.
	The finite model property is preserved for the local consequence but not for the global consequence. Moreover, the finite model property fails for the global consequence for all non-trivial fusions under both expanding and constant domain semantics.
\end{mtheorem}
The proof builds on the cactus model construction~\cite{DBLP:journals/logcom/GorankoP92}, extending it to  one‑variable modal logic via the quasimodel technique~\cite{KuruczWolterZakharyaschevGabbay2003}.  
For one‑variable first‑order modal logics with equality, the presence of equality can only have a substantial impact if one also allows non‑rigid constants or, equivalently, counting quantifiers such as $\exists^{=1}x\, P(x) \equiv \exists x\, \bigl[P(x) \wedge \forall y\, \bigl(P(y) \rightarrow (x=y)\bigr)\bigr]$, which expresses that predicate $P$ is interpreted as a singleton.
In this setting, drawing on known decidability results for the one‑variable first-order extension of the modal logic of inequality (also known as the `elsewhere' logic) and for one‑variable temporal logics~\cite{DBLP:journals/tocl/HampsonK15,DBLP:phd/ethos/Hampson16,DBLP:journals/apal/GabelaiaKWZ06} and by encoding Diophantine equations and Minsky machines to show undecidability, we establish the following non‑transfer result:
\begin{mtheorem}[Fusion with Equality]\label{thm:2}
	Decidability and recursive axiomatisability are not preserved for fusions of one-variable first-order modal logics with equality, for both the local and global consequence and under both expanding and constant domain semantics. 
	Moreover, the global consequence is undecidable for all non-trivial fusions for constant domain semantics.
\end{mtheorem}
We stress that it would be a mistake to infer that fusions of one‑variable modal logics with equality are broadly undecidable. In fact, when based on standard modal systems such as polymodal 
\textbf{K} or \textbf{S5}, one‑variable modal logics with equality remain decidable, even in \textsc{coNExpTime}~\cite{DBLP:journals/corr/abs-2509-08165}.

Since the one‑variable fragment of first‑order logic without equality can be identified with propositional modal logic \textbf{S5}, one‑variable first‑order modal logics can often be viewed as propositional multimodal logics equipped with an \textbf{S5} modality. Working within this perspective, we provide a sufficient condition for the transfer of Kripke completeness and decidability in fusions of propositional modal logics that share an \textbf{S5} modality. We say that a propositional modal logic admits $E$-homogeneous models for an equivalence relation $E$ interpreting an \textbf{S5} modality $\Box_{E}$ if, for sufficiently large infinite cardinals~$\kappa$, it is complete for Kripke models in which, for every $E$‑equivalence class $W$ of worlds, each formula is either not satisfied anywhere in $W$ or is satisfied in $\kappa$-many worlds of $W$. 
\begin{mtheorem}[Propositional Fusion with Shared \textbf{S5}]\label{thm:3}
	Kripke completeness and decidability are preserved for fusions of propositional modal logics admitting $E$-homogeneous models for a shared {\bf S5} modality~$\Box_{E}$, for both the local and global consequence.
\end{mtheorem}
The conditions of Theorem~\ref{thm:3} hold, for example, for Kripke complete  (semi)commutators with \textbf{S5} and for (expanding) product modal logics with \textbf{S5}. Hence, it provides
transfer results for corresponding one-variable modal logics. We note that,
in contrast to the proof of Theorem~\ref{thm:1} that uses quasimodels, one cannot directly prove transfer of the finite model property based on homogeneous models. Theorem~\ref{thm:3} addresses a problem posed by Baader, Ghilardi and Tinelli~\cite{DBLP:journals/iandc/BaaderGT06}.

% !TEX root =  fusions-submission.tex

\section{Preliminaries}
\label{sec:prelim}

We consider first-order (multi-)modal logics in languages with a single, fixed,
individual variable.  A~language $\mathcal{L}_M$ contains countably many unary
predicate letters, the individual variable $x$, the Booleans $\neg$ and~$\land$,
the quantifier symbol $\exists$, and a finite set $M$ of unary (necessity-like)
modal operators, denoted by~$\Box$, possibly with subscripts. Thus, in
$\mathcal{L}_M$, atomic formulas have the form $P(x)$, where $P$ is a predicate
letter.  Formulas of $\mathcal{L}_M$, denoted by $\vp, \psi, \ldots$, are built
from atomic ones in the usual way, the only allowed quantifier
being~$\exists x$.  The symbols $\bot$, $\to$, $\dis$, $\forall$, and
$\Diamond$ are standard abbreviations. We also denote~$\Box^{\leq n} \varphi = \vp \land \Box\vp \land \dots \land \Box^n\varphi$. Proposition letters~$p, q,\dots$ are
abbreviations for $\exists x\, P(x)$, $\exists x\, Q(x)$, etc. We also consider
languages with equality: $\mathcal{L}^=_M$ contains, in addition to
the symbols of $\mathcal{L}_M$, the binary predicate symbol $=$ and constant
symbols, denoted by $c_1, c_2,\dots$ Thus, in $\mathcal{L}^=_M$, a term $t$ is
either $x$ or a constant $c_i$, and atomic formulas have the form $P(t)$ or
$t_1 = t_2$.  We usually identify languages with sets of their formulas.  The
set of subformulas of a formula $\vp$ is denoted by $\mathop{\mathit{sub}} \vp$.

By a \dfnotion{one-variable \rmbr{first-order} $M$-logic}, we mean a set of
\mbox{$\mathcal{L}_M$-formulas} that includes the one-variable fragment of the
classical first-order logic $\mathbf{FO}$ and formulas
\mbox{$\Box \bigl(P(x) \to Q(x)\bigr) \to \bigl(\Box P(x) \to \Box Q(x)\bigr)$}, for all
$\Box \in M$, and is closed under Modus Ponens (MP), Generalisation
($\vp / \forall x\, \vp$), Necessitation ($\vp / \Box \vp$, for each
$\Box \in M$), and Substitution (of arbitrary one-variable formulas for formulas
of the form $P(x)$).\footnote{One-variable fragments of first-order modal logics
  are one-variable logics in our sense, but not every one-variable logic is a
  one-variable fragment of a first-order modal logic; see, e.g.,~\cite[Section
  3]{Shehtman2023}.}  If~$L$ is a one-variable $M$-logic and
$\Gamma \subseteq \mathcal{L}_M$, then $L + \Gamma$ denotes the smallest
one-variable $M$-logic including~\mbox{$L \cup \Gamma$}; we write $L + \vp$ for
$L + \{ \vp \}$.  We note that every one-variable $M$-logic contains the
converse Barcan formula
$\textit{CBF}_\Box = \Box \forall x\, P(x) \to \forall x\, \Box P(x)$, for each
$\Box \in M$.

Let $L$ be a one-variable $M$-logic and $\varphi,\psi\in \mathcal{L}_{M}$. We write $\varphi \vdash_L \psi$ and say that $\psi$ is a \dfnotion{local consequence of $\varphi$ in $L$}  if $\psi$ belongs to the smallest
set including $L \cup \{\varphi\}$ and closed under MP and Generalisation. By the deduction theorem, $L$ is decidable iff 
$\vdash_{L}$ is decidable. We write $\varphi \vdash_L^\ast \psi$ and say that $\psi$ is a \dfnotion{global consequence of $\varphi$ in~$L$} if $\psi$ belongs to the smallest set including $L \cup \{\varphi\}$ and closed under MP, Generalisation, and Necessitation
for all $\Box \in M$. We call $L$ \dfnotion{globally decidable} if 
	the global consequence in $L$ is decidable.

The \dfnotion{fusion} of a one-variable $M_1$-logic $L_1$ and a one-variable
$M_2$-logic $L_2$, where $M_1 \cap M_2 = \emptyset$, is the
smallest one-variable $(M_1 \cup M_2)$-logic $L_1 \otimes L_2$ that includes
$L_1 \cup L_2$.

We are mostly concerned with logics defined in terms of Kripke semantics.
A~\dfnotion{Kripke $M$-frame} is a tuple
$\mathfrak{F} = ( W, \{ R_\Box \}_{\Box \in M } )$, where $W$ is a non-empty
set of \dfnotion{worlds} and $R_\Box \subseteq W \times W$, for all $\Box\in M$.  We write $R_\Box(w)$ for $\{ w'\in W\mid wR_\Box w'\}$.
We say that $\mathfrak{F}$
is
\emph{non-trivial} if there are distinct $w, v \in W$ with 
$w R_{\Box} v$, for some $\Box \in M$.
A \dfnotion{Kripke
  $\textit{cd}$-model} based on $\mathfrak{F}$ is a pair
\mbox{$\mathfrak{M} = ( \mathfrak{F}, I )$}, where
\mbox{$I = \{ \lb \Delta, I_w \rb \mid w \in W \}$} is a family of classical
first-order models.  A \dfnotion{Kripke $\mathit{xd}$-model} based on $\mathfrak{F}$
is a pair $\mathfrak{M} = ( \mathfrak{F}, I )$, where
\mbox{$I = \{ \lb \Delta_w, I_w \rb \mid w \in W \}$} is a family of classical
first-order models satisfying the \dfnotion{expanding domains condition}:
$\Delta_w \subseteq \Delta_v$, for all $w, v \in W$ with~$w R_\Box v$ and
all $\Box \in M$. For languages with equality and constants, the classical
models associated with worlds are required to be \emph{normal} (i.e., $=$ is
interpreted as identity); the resultant Kripke models are called \emph{Kripke
  $d$-models with equality}, for $d\in\{\textit{xd}, \textit{cd}\}$; note that rigid interpretation of constants, i.e., $I_w (c_k) = I_v(c_k)$, for all
$w, v \in W$, is \emph{not} assumed. 
%\footnote{For results in this paper, we could, instead of $=$ and
%  non-rigid contstants, have used the quatifier ``there exists exactly one~$x$
%  such that \ldots".}  
We say that a Kripke model is \emph{finite} if its set of
worlds is finite and every classical model associated with its worlds is finite.

Since $\textit{cd}$-models are a special case of $\textit{xd}$-models, we define
truth and validity only for the latter.  Let
$\mathfrak{M} = ( W, \{ R_\Box \}_{\Box \in M}, \lb \Delta_w, I_w \rb_{w \in
  W} )$ be an $\textit{xd}$-model.  An \dfnotion{assignment} in $\mathfrak{M}$
is a map $\alpha$ sending $x$ to an element of the set
$\bigcup_{w \in W} \Delta_w$.  The \dfnotion{truth} of one-variable
$\mathcal{L}_M$- and $\mathcal{L}^=_M$-formulas at world $w$ of $\mathfrak{M}$ under assignment~$\alpha$ is
defined by recursion on the formula construction in the standard way, in particular: 
\begin{itemize}
%\item $\mathfrak{M},w \models^\alpha P(t)$ if $|t|_w^\alpha \in I_w(P)$;
%\item  $\mathfrak{M},w\models^\alpha \neg\vp$ if $\mathfrak{M},w\not\models^\alpha \vp$;
%\item  $\mathfrak{M},w\not\models^\alpha \bot$;
%\item $\mathfrak{M},w\models^\alpha \varphi_1\to\varphi_2$ if
%  $\mathfrak{M},w \not\models^\alpha \varphi_1$ or
%  $\mathfrak{M},w\models^\alpha \varphi_2$;
%\item $\mathfrak{M},w\models^\alpha \varphi_1\land\varphi_2$ if
%  $\mathfrak{M},w \models^\alpha \varphi_1$ and
%  $\mathfrak{M},w\models^\alpha \varphi_2$;
%\item $\mathfrak{M},w\models^\alpha \Box \varphi$ if
%  $\mathfrak{M},w' \models^\alpha \varphi$, for every $w' \in R_\Box(w)$;
%\item $\mathfrak{M},w \models^\alpha \forall x\,\varphi$ if
%  $\mathfrak{M},w\models^{\alpha'}\varphi$, for every $\alpha'$ such that
%  $\alpha'(x)\in \Delta_w$.
\item $\mathfrak{M},w \models^\alpha \exists x\,\varphi$ if
  $\mathfrak{M},w\models^{\alpha'}\varphi$, for some $\alpha'$ such that
  $\alpha'(x)\in \Delta_w$;
\item $\mathfrak{M},w \models^\alpha t_1 = t_2 $ if
  $|t_1|_w^\alpha = |t_2|_w^\alpha$, where $|t|_w^\alpha$ denotes the value of the term $t$ at $w$ under $\alpha$.
\end{itemize}
We say that $\vp$ is \dfnotion{true at} $w \in W$, and write
$\mathfrak{M}, w \models \vp$, if $\mathfrak{M}, w \models^\alpha \vp$ holds for
every $\alpha$ with $ \alpha(x) \in \Delta_w$. We also say that $\vp$ is
\dfnotion{true in $\mathfrak{M}$}, and write $\mathfrak{M} \models \vp$, if
$\mathfrak{M}, w \models \vp$ holds for every $w \in W$.  

Let
$d \in \{ \mathit{xd}, \mathit{cd} \}$. We write $\mathfrak{F}\models_d \varphi$ if $\varphi$ is true 
in every 
$d$-model based on~$\mathfrak{F}$. Then $\mathfrak{F}$ is a \emph{$d$-frame for} a logic $L$, in symbols $\mathfrak{F}\models_d L$, if $\mathfrak{F}\models_d \varphi$ for all $\varphi\in L$.
An $\mathcal{L}_M$-formula $\vp$ is
\dfnotion{$d$-valid} on a class $\mathscr{C}$ of Kripke $M$-frames (written as
$\mathscr{C} \models_{d} \vp$) if 
$\mathfrak{F}\models_d\vp$, for all 
$\mathfrak{F}\in\mathscr{C}$.  Analogous notions are defined for formulas and models with
equality.  We define the
following sets of formulas:
\begin{equation*}
  \Logd \mathcal{C} \  = \  \bigl\{ \vp\in\mathcal{L}_M \mid \mathcal{C} \models_d \vp \bigr\}
  \quad\text{ and }\quad
      \Loged \mathcal{C} \ = \  \bigl\{ \vp\in\mathcal{L}^=_M \mid \mathcal{C} \models_d \vp \bigr\}.
\end{equation*}
Note that, for every $\mathscr{C}$, we have
$\mathscr{C} \models_{\textit{cd}} \forall x\, \Box P(x) \to \Box \forall x\,
P(x)$ (the Barcan formula, $\textit{BF}_\Box$), for all $\Box \in M$.

Let $L$ be a one-variable $M$-logic, $\mathcal{C}$ a class of $M$-frames, and
$d \in \{\textit{xd}, \textit{cd} \}$.  We say $L$ is
\dfnotion{\mbox{$d$-complete} for~$\mathcal{C}$} if
$L = \mathop{\mathit{Log}}_{d} \mathscr{C}$, and that $L$ is \dfnotion{globally
  $d$-complete for $\mathcal{C}$} in case $\varphi\vdash_L^*\psi$ iff
$\mathfrak{M}\models\varphi$ implies $\mathfrak{M}\models\psi$, for every
$d$-model $\mathfrak{M}$ based on a frame from $\mathcal{C}$.  Note that global
$d$-completeness implies $d$-completeness, but the converse is not always true.
We say that $L$ has the \emph{finite $d$-model property \rmbr{$d$-fmp}} if, for every~\mbox{$\vp \notin L$}, there exists a finite $d$-model~$\mathfrak{M}$ such that
$\mathfrak{M} \models L$ and $\mathfrak{M} \not\models \vp$.  We say that $L$
has the \emph{global finite $d$-model property \rmbr{global $d$-fmp}} if, for every $\vp$ and $\psi$ with
$\varphi \not\vdash^*_L \psi$, there exists a finite $d$-model~$\mathfrak{M}$ such that
$\mathfrak{M} \models L \cup \{\varphi\}$ and $\mathfrak{M} \not\models \psi$.

The \dfnotion{fusion} of Kripke frames $\mathfrak{F}_1 = ( W, \{R_\Box\}_{\Box\in M_1} )$ and
$\mathfrak{F}_2 = ( W, \{R_\Box\}_{\Box\in M_2})$, where $M_1 \cap M_2 = \emptyset$, is the frame
$\mathfrak{F}_1 \otimes \mathfrak{F}_2 = ( W, \{ R_\Box\}_{\Box\in M_1 \cup M_2} )$.  
If $\mathscr{C}_1$ and
$\mathscr{C}_2$ are classes of 
Kripke frames, then 
$\mathscr{C}_1 \otimes \mathscr{C}_2$ denotes the class of frames comprising $\mathfrak{F}_1 \otimes \mathfrak{F}_2$, for all
$\mathfrak{F}_1 \in \mathscr{C}_1$ and $\mathfrak{F}_2 \in \mathscr{C}_2$ with the same set of worlds.

We will also consider propositional modal logics that can be viewed as fragments of the one‑variable modal logics introduced above. Our notation and definitions are standard, following the conventions established earlier; see also the full version~\cite{arxiv}.
The \emph{fusion} of a propositional $M_1$-logic $L_1$ and a propositional
$M_2$-logic $L_2$, where $M_1 \cap M_2 = \emptyset$, is the smallest propositional
$(M_1 \cup M_2)$-logic $L_1 \otimes L_2$ that includes~$L_1 \cup L_2$.
The
\emph{semicommutator}\footnote{In general, semicommutators are also required to contain
  confluence axioms for commuting modalities; these are derivable in
  semicommutators with $\mathbf{S5}$.}  and the \emph{commutator}~\cite{KuruczWolterZakharyaschevGabbay2003,Kurucz2007Combining} of a propositional \mbox{$M$-logic} $L$ with the
propositional monomodal logic $\mathbf{S5}$ are the propositional
\mbox{$(M \cup \{ \Box_E \})$-logics}
\begin{equation*}
 [L , \textbf{S5}]^{\textit{xd}} \ \ = \ \ L \otimes \textbf{S5}\  +\  \bigl\{
  (\textit{lcom}_{\Box}) \mid \Box \in M \bigr\} \quad\text{ and }\quad
[L, \textbf{S5}]^{\textit{cd}} \ \ = \ \ [L , \textbf{S5}]^{\textit{xd}} \ + \ \bigl\{
  (\textit{rcom}_{\Box}) \mid \Box \in M \bigr\},
\end{equation*}
where $(\textit{lcom}_{\Box})$ and $(\textit{rcom}_{\Box})$ 
are the following  formulas with their Kripke frame correspondents:
\begin{align}
 \tag{$\textit{lcom}_{\Box}$} & \Box \Box_E\, p \to \Box_E \Box\, p, && E\circ R\subseteq R\circ E,\\[2pt]
  \tag{$\textit{rcom}_{\Box}$} & \Box_E \Box\, p \to \Box \Box_E\, p, && R\circ E\subseteq E\circ R.
\end{align}

There exists a well-known translation of one-variable $M$-formulas into
propositional $(M \cup \{ \Box_E \})$-for\-mulas: $P(x)^\ast = p$;
$(\neg\vp)^\ast = \neg\vp^\ast$; $(\vp_1 \land \vp_2)^\ast = \vp_1^\ast \land \vp_2^\ast$;
$(\Box \vp)^\ast = \Box \vp^\ast$, for all $\Box \in M$;
$(\exists x\, \vp)^\ast = \Diamond_{E} \vp^\ast$.  We note that
$\textit{CBF}_\Box^\ast = (\textit{lcom}_{\Box})$ and
$\textit{BF}_\Box^\ast = (\textit{rcom}_{\Box})$. If $\Gamma \subseteq \mathcal{L}_M$,
then $\Gamma^\ast = \{ \vp^\ast \mid \vp \in \Gamma\}$.

If $L$ is a propositional $M$-logic, then the smallest one-variable $M$-logic
containing all the substitution instances of formulas from $L$ (here, we
substitute arbitrary one-variable formulas for proposition letters) is denoted
by $\mathcal{Q} L$.  
\begin{lemma}
  \label{lem:one-var-vs-prop}
  If $L$ is a propositional $M$-logic, then
  $\mathcal{Q} L^\ast = [L, \mathbf{S5}]^{\mathit{xd}}$ and
  $(\mathcal{Q} L + \mathit{BF})^\ast = [L, \mathbf{S5}]^{\mathit{cd}}$.
\end{lemma}
\begin{proof}
Let $L' = [L, \textbf{S5}]^{\textit{xd}}$ and let $\textbf{FO}_1$ be the
one-variable fragment of the classical first-order logic.  It is easy to check
that
$\vdash_{L'} \Diamond \Box_E p \to \Box_E \Diamond p (= \textit{chr}_\Diamond)$,
for each $\Diamond$ corresponding to $\Box \in M$.  We also rely on the
well-known~\cite{Wajsberg33} fact that $\textbf{FO}_1^\ast = \textbf{S5}$.

First, we prove that the translation $\cdot^\ast$ commutes with substitution:
for all one-variable formulas $\vp, \psi_1, \ldots, \psi_n$,
$$
\begin{array}{lcl} (
  \vp)^\ast\,[\psi_1, \ldots, \psi_n / P_1 (x), \ldots, P_n (x)] & = & 
  \vp^\ast\,[\psi_1^\ast, \ldots, \psi_n^\ast / p_1, \ldots, p_n].
\end{array}
\eqno{(\ast)}
$$
The proof is by induction on $\vp$; we only consider the case
$\vp = \exists x\, \theta$.  Let
$\sigma = [\psi_1, \ldots, \psi_n / P_1 (x), \ldots, P_n (x)]$ and
$\sigma^\ast = [\psi_1^\ast, \ldots, \psi_n^\ast / p_1, \ldots, p_n]$.  Since, by IH,
$(\theta\,\sigma)^\ast = \theta^\ast\,\sigma^\ast$, it follows that
$$
\begin{array}{ccccccccccccc}
  ((\exists x\, \theta)\,\sigma)^\ast & = & (\exists x\, (\theta\,\sigma))^\ast & = &
  \Diamond_E (\theta\,\sigma)^\ast & = & \Diamond_E ( \theta^\ast\,\sigma^\ast) & = &   
  (\Diamond_E \theta^\ast)\,\sigma^\ast & = &    (\exists x\, \theta)^\ast\,\sigma^\ast.
\end{array}
$$

We now prove that $L' \subseteq \mathcal{Q} L^\ast$.  It should be clear that
$L \subseteq \mathcal{Q} L^\ast$, and that
$\textbf{S5} = \textbf{FO}_1^\ast \subseteq \mathcal{Q} L^\ast$.  Since
$\textit{CBF}_\Box$ belongs to every one-variable modal logic and
$\textit{CBF}_\Box^\ast = (\textit{lcom}_\Box)$, it follows that
$(\textit{lcom}_\Box) \in \mathcal{Q} L^\ast$.  Now, the claim follows from
($\ast$) and the closure of $\mathcal{Q} L^\ast$ under MP, Necessitation, and
Generalisation.

We next prove that $\mathcal{Q} L^\ast \subseteq L'$.  Here, we use the fact
that $\mathbf{FO}_1$ is obtained by adding to the classical propositional logic
the axiom $\forall x\, P(x) \to P(x)$ and the rule
$\psi \to \vp(x) / \psi \to \forall x \, \vp(x)$, where $x$ is not free in
$\psi$.  In view of ($\ast$), the only non-trivial part is to show that
$\vdash_{L'} (\psi \to \vp)^\ast$ implies
$\vdash_{L'} (\psi \to \forall x \, \vp)^\ast$ provided that $x$ is not free in $\psi$.
To this end, we show that, if $\gamma$ is a propositional
$(M \cup \{ \Box_E \})$-formula and $\beta$ is a propositional
$(M \cup \{ \Box_E \})$-formula such that all occurrences of proposition letters
occur within the scope of $\Box_E$, then $\vdash_{L'} \beta \to \gamma$ implies
 $\vdash_{L'} \beta \to \Box_E \gamma$.  Since $\Box_E$ is monotonic, it is
enough to show that $\vdash_{L'} \beta \to \Box_E \beta$, for all $\beta$ of the
assumed form.

Due to our assumption about its form, we may rewrite $\beta$ with $\neg$,
$\wedge$, $\vee$, $\Box_E$, $\Box$, and $\Diamond$ so that all negations occur
either within the scope of $\Box_E$ or immediately precede $\Box_E$ (to that
end, we push all the negations outside of any $\Box_E$ down until they reach
some $\Box_E$).  Thus, $\beta$ is built from formulas of the form
$\Box_E \delta$ and $\neg \Box_E \delta$ using $\wedge$, $\vee$, $\Box$, and
$\Diamond$.  We proceed by induction on the structure of such $\beta$.
\begin{itemize}
\item If $\beta = \Box_E \gamma$ or $\beta = \neg \Box_E \gamma$, then the claim holds
because $\Box_E$ is an $\textbf{S5}$ modality, and therefore
$\vdash_{L'} \Box_E \gamma \to \Box_E \Box_E \gamma$ and
\mbox{$\vdash_{L'} \neg \Box_E \gamma \to \Box_E \neg \Box_E \gamma$}.  
\item The cases
of $\wedge$ and $\vee$ are straightforward.

\item If $\beta = \Box \gamma$ and $\vdash_{L'} \gamma \to \Box_E \gamma$, then, by
monotonicity of $\Box$ and $(\textit{lcom}_\Box)$,
$\vdash_{L'} \Box \gamma \to \Box_E \Box \gamma$.

\item If $\beta = \Diamond \gamma$ and $\vdash_{L'} \gamma \to \Box_E \gamma$, then, by
monotonicity of $\Diamond$ and $(\textit{chr}_\Box)$,
$\vdash_{L'} \Diamond \gamma \to \Box_E \Diamond \gamma$.
\end{itemize}
To prove the claim of the lemma for $[L, \textbf{S5}]^{\textit{cd}}$, it suffices to recall
that $\textit{BF}_\Box^\ast = (\textit{rcom}_\Box)$.
\end{proof}

%%% Local Variables:
%%% mode: latex
%%% TeX-master: "fusions-main"
%%% End:

% !TEX root =  fusions-submission.tex

\section{Fusions of Equality-Free One-Variable Logics}\label{sec:equality-free}
In this section, we prove a precise version of Theorem~\ref{thm:1}.
\begin{theorem}\label{thm:transfer:equality-free}
Let $\mathcal{C}_1, \mathcal{C}_2$ be non-empty classes of frames closed under disjoint unions, $d \in \{\textit{xd}, \textit{cd} \}$, and  $L_i$
be one-variable modal logics
\textup{(}globally\textup{)} $d$-complete for $\mathcal{C}_i$, for both $i = 1,2$. Then 
\begin{itemize}
\item $L_1\otimes L_2$ is \textup{(}globally\textup{)} $d$-complete for $\mathcal{C}_1\otimes\mathcal{C}_2$\textup{;}
\item $L_1\otimes L_2$ is \textup{(}globally\textup{)} decidable if both $L_1$ and $L_2$ are  \textup{(}globally\textup{)} decidable\textup{;}
\item $L_1\otimes L_2$ has the $d$-fmp if both $L_1$ and $L_2$ have the $d$-fmp. 
\end{itemize}
If $\mathcal{C}_1$ and $\mathcal{C}_2$ contain non-trivial frames, then $L_1 \otimes L_2$ does \emph{not} have the global $d$-fmp.
\end{theorem}

We begin the proof of these results by defining a general structure of frames for fusions, which follows Goranko and Passy's cactus model construction~\cite{DBLP:journals/logcom/GorankoP92}.
\begin{definition}
An \emph{$i$-cactus frame $(W, R_1, R_2)$ with thorns $V\subseteq W$} is defined inductively.
\begin{description}
\item[basis] A $1$-frame $(W^0, R^0_1)\in\mathcal{C}_1$ gives rise to a $2$-cactus frame $(W^0, R_1^0, \emptyset)$ with thorns $W^0$. A $1$-cactus frame  is defined symmetrically.

\item[induction step] Let $(W, R_1, R_2)$ be a $2$-cactus frame with thorns $V$ and $(W^v, R^v_2)\in\mathcal{C}_2$, for $v \in V$, a collection of $2$-frames such that, for each $v\in V$, the following conditions hold:
\begin{equation}\label{eq:cactus-frame}
W\cap W^v = \{v\} \qquad\text{ and }\qquad W^v\cap W^{v'} = \emptyset, \text{ for all } v'\in V \text{ with } v'\ne v.
\end{equation}
Then $(W', R_1, R_2 \cup \bigcup\limits_{v\in V} R_2^v)$ is a $1$-cactus frame with thorns $W'\setminus W$, where
$W'  \ = \  W \  \cup \  \bigcup\limits_{v\in V} W^v$.
Given a $1$-cactus frame and a collection of $1$-frames, a $2$-cactus frame  is defined symmetrically.
\end{description}
A \emph{limit $i$-cactus frame} is obtained as the limit of the sequence of cactus frames starting from an $i$-frame. 
\end{definition}

The following property of limit $i$-cactus frames is immediate from~\eqref{eq:cactus-frame}:
\begin{proposition}\label{prop:cactus-frame}
A limit $i$-cactus frame is a frame from $\mathcal{C}_1\otimes\mathcal{C}_2$.
\end{proposition}

With these definition and property at hand, we can establish the general non-transfer of the global $d$-fmp claimed in Theorem~\ref{thm:transfer:equality-free}.
\begin{lemma}
  \label{thr:no-gfmp}
  Let $\mathcal{C}_1$ and $\mathcal{C}_2$ be classes of frames closed under
  disjoint unions containing non-trivial frames,
  $d \in \{\textit{xd}, \textit{cd} \}$, and
  $L_i = \mathop{\mathit{Log}_d} \mathcal{C}_i$, for $i = 1,2$.  Then
  $L_1\otimes L_2$ does not enjoy the global $d$-fmp.
\end{lemma}
\begin{proof}
 We define the modal operators
  \begin{equation*}
    \Diamond_p \psi  \ \ = \ \  
     \Diamond_1 ( \neg p \ \wedge \ \Diamond_2 ( p \wedge \psi)) 
                    \qquad \text{ and } \qquad
   \Box_p \psi \ \  = \ \  \neg \Diamond_p \neg \psi,                                                                       
  \end{equation*}
  where $p$ is a proposition letter not in $\psi$, and consider the formula
  (`every dog has its day')
  \begin{equation*}
    \vp \ \  =  \ \ p \ \to  \ \big[ \exists x\, \bigl( \neg Q(x) \wedge \Diamond_p Q(x)\bigr)  
  \ \wedge \ \forall x\, \bigl( Q(x) \to \Box_p Q(x)\bigr) \bigr].
  \end{equation*}
  We show that $\vp \not\vdash_{L_1\otimes L_2}^\ast \neg p$ and that this cannot
  be witnessed by a finite model.

  First, we note that, if $\mathfrak{M}$ is a model based on a frame $\mathfrak{F} = (W, R_1, R_2)$ for
  $L_1\otimes L_2$ and $w^0\in W$ such that
  $\mathfrak{M}, w^0 \models p$ and $\mathfrak{M} \models \vp$, then
  $\mathfrak{M}$ is infinite: such a model $\mathfrak{M}$ contains an infinite
  \mbox{$R_1 R_2$-path} of worlds, starting at $w^0$, with the
  ever-increasing, non-universal extension of $Q$, all of which make $p$ true.

  Next, to show that $\vp \not\vdash_{L_1\otimes L_2}^\ast \neg p$, we define an $L_1\otimes L_2$-model 
  $\mathfrak{M}$ such that $\mathfrak{M} \models \vp$, but
  $\mathfrak{M}, w^0 \not\models p$, for some world $w^0$.  Let
  $\mathfrak{G}_1 = (U_1, S_1)$ and $\mathfrak{G}_2 = (U_2, S_2)$ be non-trivial
  frames from, respectively, $\mathscr{C}_1$ and $\mathscr{C}_2$. Thus,
  $\mathfrak{G}_1$ contains worlds $w$ and $v$ such that $w \ne v$ and
  $w S_1 v$; similarly, $\mathfrak{G}_2$ contains worlds $s$ and $t$ such
  $s \ne t$ and $s S_2 t$; see Fig.~\ref{fig:dog:cactus} for a simple example. Let
  $\mathfrak{F} = (W, R_1, R_2)$ be the limit 1-cactus frame constructed from
  the sequence $\mathfrak{G}_1^i = (U_1^i, S_1^i)$, $i \in \mathbb{N}$, of
  disjoint copies of $\mathfrak{G}_1$ and the sequence
  $\mathfrak{G}_2^i = (U_2^i, S_2^i)$, $i \in \mathbb{N}$, of disjoint copies of
  $\mathfrak{G}_2$, starting with $\mathfrak{G}^0_1$; we identify $v^i$ with
  $s^i$ and $t^i$ with $w^{i+1}$, for all $i \in \mathbb{N}$, where $u^i$
  denotes the copy of $u$ form $U^i_k$, for~$i \in \mathbb{N}$.
  By Proposition~\ref{prop:cactus-frame},
  $\mathfrak{F} \in \mathscr{C}_1 \otimes \mathscr{C}_2$, and so $\mathfrak{F}$
  is a frame for $L_1\otimes L_2$.  Let $W^0 = \{w^0\}$ and, for all
  $i \in \mathbb{N}$, let $V^i = R_2(W^i)$ and $W^{i+1} = R_1(V^i)$.
\begin{figure}%
\centering%
\begin{tikzpicture}[>=latex,
nd/.style={circle,draw,inner sep=0pt,minimum size=2mm,thick,fill=white}]
\begin{scope}[rounded corners=2mm]
\draw[fill=gray!70] (-0.2,0.3) -- ++(0.4,0) -- ++(0,1.4) -- ++(-0.4,0) -- cycle;
\draw[fill=gray!70] (1.8,-0.2) -- ++(0.4,0) -- ++(0,1.4) -- ++(-0.4,0) -- cycle;
\draw[fill=gray!70] (3.8,1.3) -- ++(0.4,0) -- ++(0,1.4) -- ++(-0.4,0) -- cycle;
\draw[fill=gray!70] (-8.2,-0.2) -- ++(0.4,0) -- ++(0,1.4) -- ++(-0.4,0) -- cycle;
\draw[fill=gray!30,fill opacity=0.5] (0.2, 1.3) -- ++(0,0.4) [rounded corners=1mm]-- ++(-2,0) [rounded corners=2mm]-- ++(0,0.5) -- ++(-0.4,0) -- ++(0,-0.9) -- cycle;
\draw[fill=gray!30,fill opacity=0.5] (2.2, -0.2) -- ++(0,0.4) [rounded corners=1mm]-- ++(-2,0) [rounded corners=2mm]-- ++(0,0.5) -- ++(-0.4,0) -- ++(0,-0.9) -- cycle;
\draw[fill=gray!30,fill opacity=0.5] (1.8, 0.8) -- ++(0,0.4) [rounded corners=1mm]-- ++(2,0) [rounded corners=2mm]-- ++(0,0.5) -- ++(0.4,0) -- ++(0,-0.9) -- cycle;
\draw[fill=gray!30,fill opacity=0.5] (3.8, 2.3) -- ++(0,0.4) [rounded corners=1mm]-- ++(2,0) [rounded corners=2mm]-- ++(0,0.5) -- ++(0.4,0) -- ++(0,-0.9) -- cycle;
\draw[fill=gray!30,fill opacity=0.5] (-6.2, 0.8) -- ++(0,0.4) [rounded corners=1mm]-- ++(2,0) [rounded corners=2mm]-- ++(0,0.5) -- ++(0.4,0) -- ++(0,-0.9) -- cycle;
\end{scope}
\node[nd] (u1) at (-2,2) {};
\node[nd] (u0) at (0,1.5) {};
\node[nd] (w0) at (0,0.5) {};
\node[nd,label=right:{$w^0$}] (v0) at (2,0) {};
\node[nd,label=left:{$v^0$ \footnotesize $(= s^0)$}] (w1) at (2,1) {};
\node[nd,label=right:{$w^1$  \footnotesize$(=t^0)$}] (v1) at (4,1.5) {};
\node[nd,label=left:{$v^1$  \footnotesize $(= s^1)$}] (w2) at (4,2.5) {};
\node[nd,label=right:{$w^2$  \footnotesize$(= t^1)$}] (v2) at (6,3) {};
\node[nd,label=left:{$s$}] (wf1) at (-6, 1) {}; 
\node[nd,label=right:{$t$}] (vf1) at (-4, 1.5) {};
\node[nd,label=left:{$v$}] (wf2) at (-8, 1) {}; 
\node[nd,label=right:{$w$}] (vf2) at (-8, 0) {}; 
\begin{scope}[thick,->,rounded corners=2mm]\scriptsize
\draw[densely dotted] (wf1) -|  (vf1);  % node[sloped,below] {$R_1$}
\draw (vf2) --  (wf2); % node[left] {$R_2$}
\draw[densely dotted] (u0) -|  (u1);  % node[sloped,below] {$R_1$}
\draw (w0) --  (u0); % node[left] {$R_2$}
\draw[densely dotted] (v0) -|  (w0);  % node[sloped,below] {$R_1$}
\draw (v0) --  (w1); % node[left] {$R_2$}
\draw[densely dotted] (w1) -|   (v1); %  node[sloped,below] {$R_1$}
\draw (v1) -- (w2); %  node[left] {$R_2$}
\draw[densely dotted] (w2) -| (v2); %  node[sloped,below] {$R_1$}
\draw (v2) -- +(0,0.5); % node[left] {$R_2$} 
\draw (u1) -- +(0,0.5); % node[left] {$R_2$} 
\end{scope}
\node at (6,3.6) {$\dots$};
\node at (-2,2.6) {$\dots$};
\node at (-5,0.4) {$\mathfrak{G}_2 = (U_2, S_2)$};
\node at (-8,1.6) {$\mathfrak{G}_1 = (U_1, S_1)$};
\node at (0.5,3) {$\mathfrak{F} = (W, R_1, R_2)$};
\end{tikzpicture}
\caption{Sample non-trivial frames $\mathfrak{G}_1$ and $\mathfrak{G}_2$ and resulting limit $1$-cactus frame $\mathfrak{F}$: solid arrows with dark background for $R_1$/$S_1$-components and dotted arrows with light background for $R_2$/$S_2$-components.}\label{fig:dog:cactus}
\end{figure}
Now, define a model
$\mathfrak{M} = ( \mathfrak{F}, \langle \mathbb{N}, I_w \rangle_{w \in W})$ by
putting $I_{w} (Q) = \{ j \in \mathbb{N} \mid j < i \}$ if $w \in W^i$, and
$I_{w} (Q) = \emptyset$ otherwise, and by making $p$ true at $w \in W$ iff
$w \in \bigcup_{i\in \mathbb{N}} W^i$.  Then $\mathfrak{M}\models \vp$ but
$\mathfrak{M}, w^0 \not\models \neg p$.
\end{proof}

The proofs of the remaining transfer results claimed in Theorem~\ref{thm:transfer:equality-free} lift the cactus model construction to the one-variable modal logic using the quasimodel technique; see, e.g.,~\cite{KuruczWolterZakharyaschevGabbay2003}.
We require two additional main ingredients: (\emph{a}) types and quasistates for a finite representation of any number of domain elements and (\emph{b}) surrogates for separating the two components of the fusion. 

In the sequel, we denote the language of $L_1 \otimes L_2$ by $\mathcal{L}$ and the language of $L_i$ by $\mathcal{L}_i$, for $i = 1, 2$. 
 For simplicity,
 we assume that $L_1$ and $L_2$ are monomodal logics with modalities $\Box_1$ and~$\Box_2$, respectively. 
Let~$\FL$ be a finite set of $\mathcal{L}$-formulas closed under subformulas. 
A \emph{$\FL$-type} $\type$ is a maximal Boolean-consistent subset of $\{ \psi, \neg\psi \mid \psi \in \FL \}$, that is, 
$\neg\psi \in \type$ iff $\psi\notin\type$, for each $\psi\in\FL$, and
$\psi_1\land\psi_2 \in \type$ iff $\psi_1,\psi_2\in\type$,  for each $\psi_1\land\psi_2\in\FL$.
The number of $\FL$-types is bounded by $2^{|\FL|}$ (for $\neg \psi\in\FL$, a type contains either $\neg \psi$ or both $\neg\neg \psi$ and $\psi$). A \emph{$\FL$-quasistate} is a non-empty set $\qs$  of $\FL$-types such that
\begin{equation*}%\label{eq:exists:saturated}
%\tag{\textbf{qs}}
\exists x\,\psi\in\type \quad\text{ iff }\quad \text{there is } \type'\in\qs \text{ with } \psi\in\type',\qquad\qquad \text{ for each }\exists x\,\psi\in\FL \text{ and } \type\in\qs.
\end{equation*}
With a $\FL$-quasistate $\qs$, we associate the following \emph{realisability $\mathcal{L}$-sentence}:
\begin{equation*}
\hat{\qs} \ \  = \ \   \forall x\,\bigvee\nolimits_{\type\in\qs} \type(x) \ \ \land \ \ \bigwedge\nolimits_{\type\in\qs} \exists x\,t(x),
\end{equation*} 
where $\type(x)$ denotes the conjunction of formulas in $\type$.
If $\QS$ is a set of $\FL$-quasistates, then $\hat{\QS}=\bigvee \{ \hat{\qs} \mid \qs\in\QS \}$.

Let $\vp$ be an $\mathcal{L}$-formula. For each $\Box_i\psi\in\sub\vp$, we introduce a fresh unary predicate letter $P_{\Box_i\psi}$ and define 
the \emph{surrogate} of $\Box_i\psi$ by taking $P_{\Box_i\psi}(x)$ if $\psi$ has a free variable and $\forall x\,P_{\Box_i\psi}(x)$ if  $\psi$ is closed.
Denote by~$\sur{1}{\vp}$
the $\mathcal{L}_1$-formula obtained from~$\vp$ by replacing $\Box_2\psi$ not in the scope of another $\Box_2$ with its surrogate; the $\mathcal{L}_2$-formula  $\sur{2}{\vp}$ is defined symmetrically.

\bigskip

As has already been observed in the propositional setting~\cite{DBLP:journals/jsyml/KrachtW91,DBLP:journals/iandc/BaaderGT06}, the analysis of local consequence is substantially more intricate than that of global consequence. We begin our proof with the latter case
and show the following lemma which states transfer of \textit{xd}-completeness of the global consequence and from which the transfer of decidability of global consequence follows by a simple enumeration of all possible $\FL$-quasistates.
\begin{lemma}\label{lemma:global}
Let $\mathcal{C}_1, \mathcal{C}_2$ be non-empty classes of frames closed under disjoint unions and  $L_i$
be one-variable modal logics globally $\textit{xd}$-complete for $\mathcal{C}_i$, for both $i = 1,2$. Let $\varphi$ and $\psi$ be $\mathcal{L}$-formulas. Denote $\FL = \sub\varphi\cup\sub\psi$. Then the following are equivalent\textup{:}
\begin{description}
\item[(G1)] $\varphi\not\vdash_{L_1\otimes L_2}^*\psi$\textup{;}
\item[(G2)] there is an $\textit{xd}$-model $\mathfrak{M}$ based on a frame from $\mathcal{C}_1\otimes\mathcal{C}_2$ such that $\mathfrak{M}\models\varphi$ but $\mathfrak{M}\not\models\psi$\textup{;}
\item[(G3)] there is a set $\QS$  of $\FL$-quasistates such that
\begin{itemize}
\item[\bf{(G3.1)}] $\sur{1}{\varphi}\land \sur{1}{\smash{\hat{\QS}}} \not\vdash_{L_1}^* \sur{1}{\psi}$,
\item[\bf{(G3.2)}] $\sur{1}{\varphi}\land \sur{1}{\smash{\hat{\QS}}} \not\vdash_{L_1}^* \neg\sur{1}{\hat{\qs}_i}$, for each $\qs_i\in\QS$,
\item[\bf{(G3.3)}] $\sur{2}{\smash{\hat{\QS}}} \not\vdash_{L_2}^* \neg\sur{2}{\hat{\qs}_i}$, for each $\qs_i\in\QS$.
\end{itemize}
\end{description}
\end{lemma}

\begin{proof}
Implication $\textbf{(G2)}\Rightarrow\textbf{(G1)}$ follows from soundness of $\vdash^*_{L_1\otimes L_2}$.

\smallskip

$\textbf{(G1)}\Rightarrow\textbf{(G3)}$ Suppose $\varphi\not\vdash_{L_1\otimes L_2}^*\psi$ and consider the following set~$\QS$ of  $\FL$-quasistates:
\begin{equation*}
\QS = \bigl\{ \qs \text{ a $\FL$-quasistate } \mid 
\varphi\not\vdash^*_{L_1\otimes L_2}\neg\hat{\qs} 
\bigr\}.
\end{equation*}
Clearly, $\bigvee_{\qs} \hat{\qs}$, for all $\FL$-quasistates $\qs$, is an $L_1\otimes L_2$-tautology, and so $\varphi\vdash^*_{L_1\otimes L_2} \bigvee_{\qs} \hat{\qs}$. Observe that $\hat{\QS}$ consists of the disjuncts left after removing from $\bigvee_{\qs} \hat{\qs}$ all $\qs$ with $\varphi\vdash^*_{L_1\otimes L_2}\neg \hat{\qs}$. Thus, we obtain \mbox{$\varphi \vdash^*_{L_1\otimes L_2} \hat{\QS}$}. 
Next, $\varphi\not\vdash^*_{L_1\otimes L_2} \psi$ implies  $\varphi\land\hat{\QS}\not\vdash^*_{L_1\otimes L_2} \psi$. Moreover, by definition, we also have \mbox{$\hat{\QS}\not\vdash^*_{L_1\otimes L_2} \neg \hat{\qs}_i$} and therefore $\varphi\land\hat{\QS}\not\vdash^*_{L_1\otimes L_2} \neg \hat{\qs}_i$,  for every $\qs_i\in \QS$.
Then~\textbf{(G3.1)}--\textbf{(G3.3)} follow by contraposition: for example,  
if we assume $\varphi\land\hat{\QS}\not\vdash^*_{L_1\otimes L_2} \psi$ but $\sur{1}{\varphi}\land \sur{1}{\smash{\hat{\QS}}} \vdash_{L_1}^* \sur{1}{\psi}$, then the latter also provides an inference in $L_1\otimes L_2$, whence, as a substitution instance, we obtain $\varphi\land\hat{\QS}\vdash^*_{L_1\otimes L_2} \psi$ contrary to our assumption.

\medskip

$\textbf{(G3)}\Rightarrow\textbf{(G2)}$ Suppose there is a set $\QS$  of $\FL$-quasistates such that \textbf{(G3.1)}--\textbf{(G3.3)} hold. 
To describe the construction of the witnessing $\textit{xd}$-model $\mathfrak{M}$, we require the following set of definitions.

\begin{definition}
A \emph{$\FL$-run} $r$ through a frame $(W, R_i)$ is a function from an $R_i$-closed subset of $W$, denoted by~$\dom r$, to~$\FL$-types. For $w\in \dom r$, a run $r$ is called
\begin{itemize}
\item \emph{$i$-coherent at $w$}  if $\Box_i \psi \in r(w)$ implies $\psi\in r(w')$,  for all $w'\in R_i(w)$;
\item \emph{$i$-saturated at $w$} if $\Box_i \psi \in r(w)$ whenever $\Box_i\psi\in\FL$ and $\psi\in r(w')$, for all $w'\in R_i(w)$. 
\end{itemize}
We say that $r$ is \emph{$i$-coherent on $W' \subseteq W$} if $r$ is $i$-coherent at each $w\in \dom r\cap W'$ and \emph{$i$-coherent} if $r$ is $i$-coherent on $W$; similarly for $r$ being $i$-saturated.

An \emph{$i$-quasimodel $\mathfrak{Q}$ for $\FL$} is a tuple of the form $((W, R_i),  w_0, \mathfrak{R})$, where  $(W, R_i)\in\mathcal{C}_i$, $w_0\in W$ 
and $\mathfrak{R}$ is a set of $\FL$-runs through $(W, R_i)$ such that each $r\in\mathfrak{R}$ is $i$-coherent and $i$-saturated, and each $\mathfrak{R}(w)$ is a $\FL$-quasistate, for $w\in W$, where we denote $\mathfrak{R}_w = \{ r\in\mathfrak{R} \mid w\in\dom r \}$ and $\mathfrak{R}(w) = \{ r(w)\mid r\in\mathfrak{R}_w \}$.
\end{definition}

\begin{definition}
An \emph{$i$-cactus $((W, R_1, R_2), w_0, \mathfrak{R})$ for $\FL$ with thorns $V\subseteq W$} is defined inductively.
\begin{description}
\item[basis] A $1$-quasimodel $((W^0, R^0_1), w_0, 
\mathfrak{R}^0)$ for $\FL$ with a $1$-frame $(W^0, R^0_1)\in\mathcal{C}_1$ gives rise to a $2$-cactus $((W^0, R_1^0, \emptyset), w_0, %\qs^0, 
\mathfrak{R}^0)$ for $\FL$ with thorns $W^0$. 
Note that each $r\in \mathfrak{R}^0$ is vacuously $2$-coherent, but not necessarily $2$-saturated. A $1$-cactus is defined symmetrically.

\item[induction step] Let $((W, R_1, R_2), w_0, 
\mathfrak{R})$ be a $2$-cactus for $\FL$ with thorns $V$ and $\mathfrak{Q}^v$, for $v \in V$, a collection of $2$-quasimodels  $\mathfrak{Q}^v = ((W^v, R^v_2), v, 
\mathfrak{R}^v)$ for $\FL$ such that, for each $v\in V$, we have~\eqref{eq:cactus-frame} and 
\begin{equation*}
\mathfrak{R}(v) = \mathfrak{R}^v(v). 
\end{equation*}
Then $((W', R_1, R_2 \cup \bigcup\limits_{v\in V} R_2^v), w_0, 
\mathfrak{R}')$ is a $1$-cactus for $\FL$ with thorns $W'\setminus W$, where $W'  =  W  \cup \bigcup\limits_{v\in V} W^v$,
\begin{equation*}
\mathfrak{R}' \ \ = \ \ \bigcup_{r\in\mathfrak{R}} \bigl\{  r \cup \hspace*{-1.2em}\bigcup_{v\in \dom r\cap V}\hspace*{-1.2em} r^v \mid    r^v\in(\mathfrak{R}^v)_v \text{ with } r(v) = r^v(v), \text{ for } v\in \dom r\cap V \bigr\} \ \ \cup \ \ \bigcup_{v\in V} (\mathfrak{R}^v \setminus (\mathfrak{R}^v)_v),
\end{equation*}
and $f\cup \bigcup_{v\in V'} f^v$ denotes the partial function that coincides with $f$ on $\dom f$ and with $f^v$ on $\dom f^v$, for each~$v\in V'$; clearly, the $\Phi$-runs in $\mathfrak{R}'$ are well-defined. We say that the $\mathfrak{Q}^v$ are \emph{grafted} onto the 2-cactus.
It can be seen that each $r'\in\mathfrak{R}'$ is $1$- and $2$-coherent as well as $2$-saturated on $W'$, and is $1$-saturated on $W$ but not necessarily on $W'\setminus W$.
Given a $1$-cactus for $\FL$ and a collection of $1$-qusimodels, a 2-cactus for $\FL$ is defined symmetrically.
\end{description}
A \emph{limit $i$-cactus for $\FL$} is obtained as the limit of the sequence of cacti starting from a $i$-quasimodel. 
\end{definition}

We construct a limit 1-cactus for $\FL$ by using the following procedure.  We first observe that, by~\textbf{(G3.2)} and global Kripke completeness of $L_1$, we have $\sur{1}{\varphi}\in\sur{1}{\type}$ and so $\varphi\in\type$ and thus $\sur{2}{\varphi}\in\sur{2}{\type}$, for every $\type\in\qs$ and $\qs\in\QS$. 
We then begin, by~\textbf{(G3.1)} and global Kripke completeness of $L_1$, with an $\textit{xd}$-model~$\mathfrak{M}^0$ based on $(W^0, R^0_1)\in\mathcal{C}_1$ such that $\mathfrak{M}^0,w\models \sur{1}{\varphi}$ and $\mathfrak{M}^0,w\models  \sur{1}{\smash{\hat{\QS}}}$, for all $w\in W^0$, but $\mathfrak{M}^0,w_0\not\models\sur{1}{\psi}$, for some $w_0\in W^0$. 
With each  $w\in W^0$ and $e\in\Delta^0_w$, we associate a uniquely determined $\FL$-type $\type_{w,e}$ such that
\begin{equation*}
\mathfrak{M}^0,w\models \sur{1}{\type_{w,e}}[e].
\end{equation*}
Let $\mathfrak{R}^0$  be the set of functions $r_e$, for $e\in\bigcup_{w\in W^0}\Delta^0_w$, that map each $w\in W^0$ with $e\in\Delta^0_w$ to the type~$\type_{w,e}$. It follows that  each $r_e$ is a  $1$-coherent and $1$-saturated $\Phi$-run. 
The $1$-quasimodel $((W^0, R^0_1), w_0, 
\mathfrak{R}^0)$  for~$\FL$ gives rise to our initial $2$-cactus for~$\FL$ with thorns~$W^0$. Observe that, for every $\chi\in\FL$  and $i =1,2$, we have
\begin{equation*}
\mathfrak{M}^0,w\models \sur{i}{\chi}[e] \qquad\text{ iff }\qquad \sur{i}{\chi}\in \sur{i}{\type_{w,e}}, \qquad\text{ for all }w\in W^0 \text{ and } e\in\Delta^0_w.
\end{equation*}

Let  $\mathfrak{C} = ((W, R_1, R_2), w_0,
\mathfrak{R})$ be a $2$-cactus for~$\FL$ with thorns $V\subseteq W$. Take any  $v\in V$ and consider the $\FL$-quasistate $\mathfrak{R}(v) = \qs^v \in \QS$ selected for~$v$. 
By~\textbf{(G3.3)} and global Kripke completeness of~$L_2$, there is an $\textit{xd}$-model $\mathfrak{M}^v$ based on $(W^v, R_2^v)\in\mathcal{C}_2$ such that  $\mathfrak{M}^v,v\models \sur{2}{(\hat{\qs}^v)}$ and $\mathfrak{M}^v, w\models  \sur{2}{\smash{\hat{\QS}}}$, for all~$w\in W^v$. 
Similarly to the previous step, with each $w\in W^v$ and $e\in\Delta_w^v$, we associate a uniquely determined $\FL$-type~$\type^v_{w,e}$ such that 
\begin{equation*}
\mathfrak{M}^v,w\models \sur{2}{(\type^v_{w,e})}[e].
\end{equation*}
Let $\mathfrak{R}^v$  be the set of functions $r^v_e$, for $e\in\bigcup_{w\in W^v}\Delta_w^v$, that map each $w\in W^v$ with $e\in \Delta^v_w$ to the type~$\type^v_{w,e}$. We have $\mathfrak{R}^v(v) = \qs^v = \mathfrak{R}(v)$. As before, $\mathfrak{Q}^v = ((W^{v}, R_2^v), v,
\mathfrak{R}^v)$, for each thorn $v\in V$, is a $2$-quasimodel for $\FL$.  We can assume that the $W^v$ satisfy~\eqref{eq:cactus-frame}, and so we extend the 2-cactus $\mathfrak{C}$ to a $1$-cactus $\mathfrak{C}'$ for $\FL$ by grafting 2-quasimodels $\mathfrak{Q}^v$ onto~$\mathfrak{C}$.

By repeating this procedure, we obtain a limit 1-cactus~$\mathfrak{C} = ((W, R_1, R_2), w_0, %\qs, 
\mathfrak{R})$ for $\FL$. Since the $\mathfrak{M}^v$ realise the respective quasimodels, for every $\chi\in\FL$,  $i = 1, 2$ and  $v\in W$, we have
\begin{equation*}
\mathfrak{M}^{v},w\models \sur{i}{\chi}[e] \
\qquad\text{ iff }\qquad \sur{i}{\chi}\in \sur{i}{(r_e(w))}, \qquad\text{ for all }w\in W^v \text{ and } e\in\Delta^{v}_w.
\end{equation*}
We now construct a Kripke model $\mathfrak{M}$ from the limit 1-cactus $\mathfrak{C}$ for $\FL$. The frame is taken from~$\mathfrak{C}$ (by assumption, both $\mathcal{C}_1$ and $\mathcal{C}_2$ are closed under disjoint unions), the domain $\Delta_w=\mathfrak{R}_w$,  for each $w\in W$, and $\mathfrak{M},w \models P[r]$ iff $P(x)\in r(w)$, for every unary predicate letter $P$, $w\in W$ and   $r\in\mathfrak{R}_w$. 
By induction on the structure of $\chi\in\FL$, we show that, for each $w\in W$ and $r\in \mathfrak{R}_w$,
\begin{equation*}
\mathfrak{M},w\models \chi[r] \qquad\text{ iff }\qquad \sur{1}{\chi}\in \sur{1}{(r(w))}\qquad\text{ iff }\qquad \sur{2}{\chi}\in \sur{2}{(r(w))}.
\end{equation*}
The basis of induction is by definition. The cases of the Boolean connectives and quantifiers follow from the definition of  quasistates. The cases of the modalities follow from the fact that the $\Phi$-runs in $\mathfrak{R}$ are $i$-coherent and $i$-saturated.
Thus, $\mathfrak{M},w\models \varphi$, for all $w\in W$, but $\mathfrak{M},w_0\not\models \psi$.  
\end{proof}

The claims in Theorem~\ref{thm:transfer:equality-free} for \textit{cd}-models are obtained by a straightforward modification of the construction in Lemma~\ref{lemma:global} (the runs are required to be total functions on the sets of worlds).

% !TEX root =  fusions-submission.tex

\bigskip

We now turn to the case of local consequence, which is technically more elaborate. The construction is also inspired by the cactus (quasi)models, but instead of starting from quasimodels and grafting more quasimodels to the cactus to obtain the required (quasi)model in the limit, the proof for the local consequence proceeds in the opposite direction: it starts from the `leaves' and grafts smaller cacti onto a quasimodel. In fact, this proof requires only a finite number of steps, which depends on the so-called alternation depth (which is formally defined below, but, intuitively, reflects the number of switches between the modalities in the formula), while the `relevant'  part of types in quasimodels gets smaller and smaller as we move away from the world where the given formula is to be satisfied. 

For an $\mathcal{L}$-formula $\varphi$, consider the following set of subformulas:
\begin{equation*}
\sur{1}{\Theta}(\varphi) = \{ P(x) \mid P(x) \in \sub\varphi \} \cup \{ \chi \in \sub \Box_2\psi \mid \Box_2\psi \in \sub\varphi\}.
\end{equation*}
It can be thought of as %all `atomic' subformulas that have all the outermost $\Box_1$-operators removed. 
removing the outermost `layer' of Booleans, quantifiers and $\Box_1$-operators.
The set $\sur{2}{\Theta}(\varphi)$ is defined symmetrically. For a set $\Phi$ of $\mathcal{L}$-formulas, let $\sur{i}{\Theta}(\Phi) = \bigcup_{\varphi\in\Phi} \sur{i}{\Theta}(\varphi)$.

We now define the notion of \emph{alternation depth} of a formula using a converse operation to $\sur{i}{\Theta}$. More precisely, for an $\mathcal{L}$-formula $\psi$ with surrogates for $\Box_2$-subformulas, we denote
\begin{equation*}
\rsur{1}{\psi} = \psi[ \sur{1}{(\Box_2\chi)}/\sur{2}{(\Box_2\chi)} \mid  \sur{1}{(\Box_2\chi)} \in \sub\psi].
\end{equation*}
%
%{\color{red} should be $\rsur{1}{\sur{1}{\varphi}}$?}
In other words, operation $\rsur{1}{\cdot}$ restores the `meaning' of  the surrogates by re-introducing their $\Box_2$-modal\-ities, but the replacement subformulas are built from surrogates for their own $\Box_1$-subformulas instead (and so contain no $\Box_1$-modalities themselves); $\rsur{2}{\psi}$ is defined symmetrically.  It should be clear that, if these operations are applied repeatedly, alternating between $\rsur{1}{\cdot}$ and $\rsur{2}{\cdot}$, then we arrive at the original formula without surrogates: more precisely, for any $\mathcal{L}$-formula $\varphi$, we have $\varphi = \rsur{i}{\rsur{2}{\rsur{1}{(\sur{1}{\varphi})}}\cdots}$, for some
%a sequence of alternating $\rsur{1}{\cdot}$ and $\rsur{2}{\cdot}$ operations and 
$i\in\{1,2\}$. Denote by $\sur{1}{\adp}(\varphi)$ the minimum length of such a sequence; symmetrically,  $\sur{2}{\adp}(\varphi)$ is the smallest number of $\rsur{2}{\cdot}$ and $\rsur{1}{\cdot}$ operations such that $\varphi = \rsur{i}{\rsur{1}{\rsur{2}{(\sur{2}{\varphi})}}\cdots}$. For example, $\sur{1}{\adp}{(\Box_1\Box_2 p)} = 1$ and $\sur{1}{\adp}{(\Box_2\Box_2\Box_1r)} = 2$; thus, $\sur{1}{\adp}(\Box_1\Box_2 p \land \Box_2\Box_2\Box_1 r) = 2$.  
It follows from the definitions that
\begin{equation*}
|\sur{1}{\adp}{(\varphi)} - \sur{2}{\adp}{(\varphi)}| \leq 1, 
\end{equation*}
depending on the leads of the longest alternation sequences in $\varphi$. Denote $\sur{i}{\adp}(\Phi) = \sur{i}{\adp}(\bigwedge\nolimits_{\varphi\in\Phi}\varphi)$.
Let $\adp(\varphi) = \min(\sur{1}{\adp}{(\varphi)}, \sur{2}{\adp}{(\varphi)})$ and $\adp(\Phi) = \adp(\bigwedge\nolimits_{\vp\in\Phi}\vp)$.  

Observe that,  for any set $\Phi$ of $\mathcal{L}$-formulas and $i= 1,2$ with $\adp(\Phi) = \sur{i}{\adp}(\Phi)$, we have
\begin{equation}
\tag{$\dagger$}
\label{eq:adp}
%\begin{array}{l}
\adp(\sur{i}{\Theta}(\Phi)) = \max\{0, \adp(\Phi) - 1\} \quad\text{ and }\quad \adp(\sur{i}{\Theta}(\Phi)) = \sur{3-i}{\adp}(\sur{i}{\Theta}(\Phi)).
\end{equation}

In addition to $\adp$, we
define the \emph{modal $1$-depth $\sur{1}{\md}(\varphi)$ of $\varphi$} as the maximum number of nested $\Box_1$ (ignoring $\Box_2)$; the \emph{modal $2$-depth $\sur{2}{\md}(\varphi)$} is defined symmetrically. We also set $\sur{i}{\md}(\Phi)=\sur{i}{\md}( \bigwedge\nolimits_{\varphi\in\Phi}\varphi)$. The modal $i$-depth will be used to describe the decreasing number of `relevant' formulas as we move away from the world witnessing the formula, while the alternation depth will be used to count the number of requires $R_1$- and $R_2$-components along the paths.

Finally, for an $\mathcal{L}$-formula  $\varphi$, consider
\begin{equation*}
\QS_i(\varphi) = \bigl\{ \qs \text{ a $\sur{i}{\Theta}(\varphi)$-quasistate such that }  \not\vdash_{L_1\otimes L_2}\neg\hat{\qs}\bigr\},
\end{equation*}
and, as before, let $\hat{\QS}_i(\varphi)$ denote the disjunction of all $\hat{\qs}$, for $\qs\in \QS_i(\varphi)$. 

With these definitions at hand, we can formulate our criterion for deciding local consequence, as claimed in Theorem~\ref{thm:transfer:equality-free}.

\begin{lemma}\label{lemma:local}
Let $\mathcal{C}_1, \mathcal{C}_2$ be non-empty classes of frames closed under disjoint unions and  
%$L_i$ be one-variable modal logics $\textit{xd}$-complete for $\mathcal{C}_i$, 
$L_i = \mathop{\mathit{Log}_{\textit{xd}}} \mathcal{C}_i$,
for both $i = 1,2$. For all $\mathcal{L}$-formulas $\varphi$ and $i = 1,2$ with $\adp(\varphi) = \sur{i}{\adp}{(\varphi)}$, the following are equivalent\textup{:}
\begin{description}
\item[(L1)] $\not\vdash_{L_1\otimes L_2} \varphi$\textup{;} 
\item[(L2)] there is an $\textit{xd}$-model $\mathfrak{M}$ based on a $\mathcal{C}_1\otimes\mathcal{C}_2$-frame such that $\mathfrak{M}\not\models\varphi$\textup{;}
\item[(L3)] $\not\vdash_{L_i} \smash{\Box_i^{\leq \sur{i}{\md}(\varphi)}} \sur{i}{(\hat{\QS}_i(\varphi))} \to \sur{i}{\varphi}$.
\end{description}
\end{lemma}

Completeness of $L_1\otimes L_2$ for $\mathcal{C}_1\otimes \mathcal{C}_2$ claimed in Theorem~\ref{thm:transfer:equality-free}  then follows immediately. For decidability, observe that, while the definition of $\QS_i(\varphi)$ refers to $\vdash_{L_1\otimes L_2}$ and thus may appear circular, by property~\eqref{eq:adp}, the $\sur{i}{\Theta}(\varphi)$-formulas have smaller alternation depth, and so we can apply the criterion in Lemma~\ref{lemma:local} recursively to the simpler formulas to decide whether the respective $\qs$ should be included in the set $\QS_i(\varphi)$. This recursion terminates as we eventually reach formulas of alternation depth 0, which, by definition, belong to either $\mathcal{L}_1$ or $\mathcal{L}_2$, where, by assumption, the consequence relation is decidable.

%We prove Lemma~\ref{lemma:local} for expanding domains (the case of constant domains is similar). %We require some additional definitions.

\begin{definition}
Let $\SIG$ be a set of $\mathcal{L}$-formulas closed under subformulas, For $i = 1,2$ and  $h \leq \sur{i}{\md}(\SIG)$, denote $\sur{i}{\Gamma_h}(\SIG) = \{ \psi \in \SIG \mid \sur{i}{\md}(\psi) \leq h\}$. Since $\sur{i}{\Gamma}_h(\SIG) = \emptyset$ for $h < 0$, there is only one such $\sur{i}{\Gamma}_h(\SIG)$-type, $\emptyset$. For a type $\type$, let $\type|_\SIG$ be the \emph{restriction} of $\type$ to $\SIG$, that is, $\type\cap \{ \psi, \neg\psi \mid \psi \in \SIG\}$, 
and denote $\qs|_\SIG = \{\type|_\SIG \mid \type\in\qs\}$. 

Let $(W, R_i)\in\mathcal{C}_i$. For $w, w' \in W$, let $\dist_{R_i}(w,w')$ be the length of the shortest $R_i$-path from $w$ to~$w'$. A \emph{tapered $\SIG$-run through $((W, R_i), w_0)$} is a function $r$ from an $R_i$-closed subset of $W$, denoted by~$\dom r$,  such that each $r(w)$ is a $\sur{i}{\Gamma_{h(w)}}(\SIG)$-type, where $h(w) = \sur{i}{\md}(\SIG) - \dist_{\smash{R_i}}(w_0, w)$. For $w\in \dom r$, run $r$ is
\begin{itemize}
\item \emph{$i$-coherent at $w$}  if $\Box_i \psi \in r(w)$ implies $\psi\in r(w')$,  for all $w'\in R_i(w)$;
\item \emph{$i$-saturated at $w$} if $\neg\Box_i \psi \notin r(w)$ whenever $\psi\in r(w')$, for all $w'\in R_i(w)$.\footnote{We use $\neg\Box_i \psi \notin r(w)$ as a convenient replacement for  either $\Box_i\psi\notin\sur{i}{\Gamma_{h(w)}}(\SIG)$ or $\Box_i \psi \in r(w)$.} 
\end{itemize}
We say that $r$ is \emph{$i$-coherent on $W' \subseteq W$} if $r$ is $i$-coherent at each $w\in \dom r\cap W'$ and \emph{$i$-coherent} if $r$ is $i$-coherent on $W$; similarly for $r$ being $i$-saturated.

A \emph{tapered $i$-quasimodel $\mathfrak{Q}$ for $\SIG$} is a tuple of the form $((W, R_i),  w_0, \mathfrak{R})$, where  $(W, R_i)\in\mathcal{C}_i$, $w_0\in W$, 
and $\mathfrak{R}$ is a set of tapered $\SIG$-runs through $((W, R_i), w_0)$ such that each $r\in\mathfrak{R}$ is $i$-coherent and $i$-saturated and each $\mathfrak{R}(w)$ is a $\SIG$-quasistate, for  $w\in W$, with $\mathfrak{R}_w = \{ r\in\mathfrak{R} \mid w\in\dom r \}$ and $\mathfrak{R}(w) = \{ r(w)\mid r\in\mathfrak{R}_w \}$.
\end{definition}

\begin{definition}
We inductively define a  \emph{tapered $i$-cactus $\mathfrak{C}$ for $\SIG$} as a tuple of the form $((W, R_1, R_2), w_0,\mathfrak{R})$, where  $(W, R_i)\in\mathcal{C}_i$ and %$R_2\subseteq (W\setminus \{ w_0 \})^2$, 
$(W\setminus \{w_0\}, R_{3-i})\in\mathcal{C}_{3-i}$ with a distinguished world $w_0\in W$, which is not in the domain or range of $R_{3-i}$, and $\mathfrak{R}$ is a set of tapered $\FL$-runs that are $i$- and $(3-i)$-coherent as well as $i$-saturated on~$W$ but $(3-i)$-saturated only on $W\setminus \{w_0\}$; %A tapered $2$-cactus is defined symmetrically.
see Fig.~\ref{fig:cactus}, which shows the frame structure.%
\begin{figure}%
\centering%
\begin{tikzpicture}[>=latex,
nd/.style={circle,draw,inner sep=0pt,minimum size=2mm,thick,fill=white}]
\begin{scope}[rounded corners=2mm,fill=gray!70]
\filldraw (-0.25,-0.25) rectangle ++(0.5,1.1);
\filldraw (-2.25,0.95) rectangle ++(0.5,1.1);
\filldraw (-4.25,0.95) rectangle ++(0.5,1.1);
\filldraw (3.75,0.95) -- ++(0.75,0) -- ++(0.25,0.25) -- ++(0,0.75) -- ++(-0.5,0) [rounded corners=1mm]-- ++(0,-0.5) [rounded corners=2mm]-- ++(-0.5,0) -- cycle;
\filldraw (2.25,0.95) -- ++(-0.75,0) -- ++(-0.25,0.25) -- ++(0,0.75) -- ++(0.5,0) [rounded corners=1mm]-- ++(0,-0.5)  [rounded corners=2mm]-- ++(0.5,0) -- cycle;
\end{scope}
\begin{scope}[rounded corners=2mm,fill=gray!30,fill opacity=0.5]
\filldraw (0.25,-0.25) -- ++(0,0.5) -- ++(-2,0) -- ++(0,1.2) -- ++(-2.5,0) -- ++(0,-0.5) -- ++(2,0) -- ++(0, -0.95) -- ++(0.25,-0.25) -- cycle;
\filldraw (-0.25,0.85) -- ++(2,0) -- ++(0,0.6) -- ++(2.5,0) -- ++(0,-0.5) -- ++(-2,0) -- ++(0,-0.35) -- ++(-0.25,-0.25) -- ++(-2.25,0) -- cycle;
\filldraw (-2.25,1.55) -- ++(0.75,0) -- ++(0.25,0.25) -- ++(0,0.75) -- ++(-0.5,0) [rounded corners=1mm]-- ++(0,-0.5)  [rounded corners=2mm]-- ++(-0.5,0) -- cycle;
\filldraw (-3.75,1.55) -- ++(-0.75,0) -- ++(-0.25,0.25) -- ++(0,0.75) -- ++(0.5,0) [rounded corners=1mm]-- ++(0,-0.5)  [rounded corners=2mm]-- ++(0.5,0) -- cycle;
\end{scope}
\node[nd,label=below:{\rule{0pt}{12pt}$w_0\colon \neg\Box_1\Box_2 \neg p \land \neg \Box_2\Box_2\Box_1 \neg r$}] (w0) at (0,0) {};
\node[nd,label=above:{\rule[-4pt]{0pt}{12pt}$w_1$}] (w1) at (0,0.6) {};
\draw[->,thick] (w0) -- (w1);
\node[nd,label=right:{\hspace*{4pt}$w_2$}] (w01) at (-2,1.2) {};
\node[nd,label=below:{\rule{0pt}{8pt}$w_3$}] (w02) at (-4,1.2) {};
\node[nd,label=above:{\large\rule[-4pt]{0pt}{10pt}$p$}] (w11) at (2,1.2) {};
\node[nd] (w12) at (4,1.2) {};
\begin{scope}[densely dotted,thick]
\draw[->,rounded corners=2mm] (w0) -| (w01);
\draw[->,rounded corners=2mm] (w1) -| (w11);
\draw[->] (w01) -- (w02);
\draw[->] (w11) -- (w12);
\end{scope}
\node[nd] (w011) at (-2,1.8) {};
\node[nd,label=right:{\hspace*{1mm}\large$r$}] (w021) at (-4,1.8) {};
%\node[nd] (w111) at (2,1.8) {};
%\node[nd] (w121) at (4,1.8) {};
%
\begin{scope}[thick]
\draw[->] (w01) -- (w011);
\draw[->] (w02) -- (w021);
%\draw[->] (w11) -- (w111);
%\draw[->] (w12) -- (w121);
\end{scope}
\node at (-4.5, 2.3) {$\vdots$};
\node at (-1.5, 2.3) {$\vdots$};
\node at (1.5, 1.7) {$\vdots$};
\node at (4.5, 1.7) {$\vdots$};
\node (adp2) at (5.6, -0.35) {\small$\textit{adp}(\varphi) = 2$, \footnotesize $\sur{1}{\md}(\vp) = 1$};
\node[below of=adp2,yshift=6mm] {\footnotesize $\sur{1}{\Theta}(\sub \varphi) = \{ \Box_2 \neg p,\ \Box_2\Box_2\Box_1 \neg r,\ \dots \}$};
\draw[->,thin,dashed] (adp2.west) to (0.25,0.2);
\node (adp1) at (-5.6, 0) {\small $\textit{adp}(\vp^{w_0}) = 1$, \footnotesize $\sur{2}{\md}(\vp^{w_0}) = 2$};
\node[below of=adp1,yshift=5.7mm] (adp1c) {\footnotesize $\FL^{w_0} = \sur{1}{\Theta}(\sub \varphi)$};
\node[below of=adp1c,yshift=6mm]  {\footnotesize $\sur{2}{\Theta}(\FL^{w_0}) = \{ p, \Box_1 \neg r, \dots \}$};
%\draw[->,thin,dashed,rounded corners=5mm] (adp1) to (0,1.1) to (1,0.85);
\draw[->,thin,dashed] (adp1.east) to (-2.25,0.5);
\node (adp12) at (6, 0.7) {\small $\textit{adp}(\varphi^{w_1}) = 1$, \footnotesize $\sur{2}{\md}(\vp^{w_1}) = 0$};
\node[below of=adp12,yshift=5.5mm] (adp12c) {\footnotesize $\FL^{w_1} =\ ${\scriptsize $\sur{1}{\Theta}(\sub \varphi) \cap \sur{1}{\Gamma}_{0}(\sub\varphi) =\ $}$ \{p, r \}$};
%\node[below of=adp12c,yshift=6mm]  {\footnotesize $\sur{2}{\Theta}(\FL^{w_1}) = \{ p, r, \dots \}$};
\draw[->,thin,dashed] (adp12) to (2.15,0.5);
\node (adp0) at (0.3,2.1)  {\footnotesize\tabcolsep=0pt\begin{tabular}{c}degenerate\\[-3pt]2-cacti\end{tabular}};
%\draw[->,thin,dashed] ($(adp0.south)+(0.3,0)$) to (1.75,1.4);
\draw[->,thin,dashed] (adp0) to (-1.25,2);
%\draw[->,thin,dashed,rounded corners=5mm] (adp0) to (1.7,2.8) to (3.75,1.5);
\draw[->,thin,dashed,rounded corners=5mm] (adp0) to (-1.5,2.7) to (-4.25,2.4);
\node (adp02) at (7,1.8) {\footnotesize\tabcolsep=0pt\begin{tabular}{c}degenerate\\[-3pt]1-cacti\end{tabular}};
%\node (dc) at (7,2) {\small\parbox{27mm}{\centering degenerate 1-cactus\\for $\SIG=\emptyset$}};
\draw[->,thin,dashed] (adp02) -- (4.8,1.6);
\draw[->,thin,dashed,rounded corners=5mm] (adp02) to (4.3,2.2) to (1.7, 1.85);
\node (adp3) at (-6.2, 2.3) {\small $\adp(\vp^{w_3}) = 1$};
\node[below of=adp3,yshift=0mm]  {\footnotesize\tabcolsep=0pt\begin{tabular}{c}$\sur{1}{\md}(\vp^{w_3}) = 1$\\$\FL^{w_3} =\ $\\[-2pt]\tiny $\sur{2}{\Theta}(\FL^{w_0}) \cap \sur{2}{\Gamma}_{0}(\FL^{w_0}) = \ $\\[0pt]
$\{p, \Box_1 \neg r, \dots \}$\end{tabular}};
\draw[->,thin,dashed] ($(adp3.east)+(-0.2,-1)$) -- (-4.3,1.3);
\end{tikzpicture}%
\caption{A fully grafted tapered $1$-cactus satisfying $\varphi = \Diamond_1\Diamond_2 p \land \Diamond_2\Diamond_2\Diamond_1 r$: solid arrows with dark background for $R_1$-components and dotted arrows with light background for $R_2$-components.}\label{fig:cactus}
\end{figure}
\begin{description}
\item[basis: $\SIG = \emptyset$.]  A tapered `degenerate' $1$-cactus $((W, R_1, R_2), w_0,\mathfrak{R})$ for $\emptyset$ is constructed by first taking sufficiently large disjoint unions $(W, R_1)$ and $(W', R_2')$ of copies of $\mathfrak{F}_1\in \mathcal{C}_1$ and $\mathfrak{F}_2\in \mathcal{C}_2$, respectively,  such that $|W|- 1 = |W'|$. We then take any $w_0\in W$ and obtain the required $(W, R_1, R_2)$ by applying to~$R_2'$ the bijection between $W\setminus \{w_0\}$ and $W'$.
Finally, $\mathfrak{R}$ consists of a single tapered $\emptyset$-run~$r$ with $r(w) = \emptyset$, for all~$w\in W$, which is trivially $1$- and $2$-coherent and $1$- and $2$-saturated.  

\item[induction step: $\SIG\ne\emptyset$.] Let $\mathfrak{Q} = ((W^0,R_1^0), w_0,\mathfrak{R}^0)$  be a tapered $1$-quasimodel  for $\SIG$. 
For each $v\in W^0$, denote
\begin{equation*}
\SIG^v = \sur{1}{\Theta}(\SIG) \cap \sur{1}{\Gamma}_{\smash{h(v)}}(\SIG),\quad \text{ where } h(v) = \sur{1}{\md}(\SIG) - \dist_{\smash{R_1^0}}(w_0,v).
\end{equation*}
For each $v\in W^0\setminus \{ w_0 \}$, take  a tapered $2$-cactus $\mathfrak{C}^v = ((W^v, R_1^v, R_2^v), v, \mathfrak{R}^v)$ for $\SIG^v$ that is \emph{compatible} with $\mathfrak{Q}$:
\begin{equation*}
W^0 \cap W^v = \{v\}\quad 
\text{  and }\quad \mathfrak{R}^0(v) |_{\sur{1}{\Theta}(\SIG)}  = \mathfrak{R}^v(v). 
\end{equation*}
We also assume that the $W^v$ are disjoint.  Then $((W, R_1, R_2), w_0,\mathfrak{R})$  is a \emph{tapered $1$-cactus $\mathfrak{C}$ for~$\SIG$} if it is obtained by grafting the $\mathfrak{C}^v$, for $v\in W^0\setminus\{w_0\}$, onto $\mathfrak{Q}$.
More precisely, we say that $((W^0, R_1^0, \emptyset), w_0, \mathfrak{R}^0)$ is an initial \emph{stock} and $W^0$ its thorns. Then, given a stock $((W, R_1, R_2), w_0, \mathfrak{R})$ and a thorn $v \in W$,  the \emph{grafting of $\mathfrak{C}^v$ onto $v$} is a stock $((W \cup W^v, R_1 \cup R_1^v, R_2 \cup R_2^v), w_0, \mathfrak{R}')$, where
\begin{equation*}
%W' \ \ = \ \ W \cup W^v,\qquad ,\qquad R_2' \ \ = \ \ ,\\
%\text{and }\quad 
\mathfrak{R}'  \ \ = \ \  \mathfrak{R}\setminus \mathfrak{R}_v \ \ \cup \ \ \bigl\{ r\cup r^v \mid r\in\mathfrak{R}_v \text{ and } r^v\in(\mathfrak{R}^v)_v \text{ with }  r(v)|_{\smash{\sur{1}{\Theta}(\SIG)}} =r^v(v) \bigr\} \ \ \cup \ \ \mathfrak{R}^v\setminus (\mathfrak{R}^v)_v.
\end{equation*}
%
%\begin{multline*}
%W \ \ = \ \ \{w_0 \} \cup \bigcup\nolimits_{v\in W^0\setminus\{w_0\}} W^v,\qquad R_1 \ \ = \ \ R_1^0 \ \cup \ \bigcup\nolimits_{v\in W^0\setminus\{w_0\}} R_1^v,\qquad R_2 \ \ = \ \ \bigcup\nolimits_{v\in W^0\setminus \{w_0\}} R_2^v,\\
%\text{and }\quad \mathfrak{R}'  \ \ = \ \  \bigl\{ r\cup \bigcup\nolimits_{v\in W^0\setminus\{w_0\}} r^v \mid r\in\mathfrak{R}^0 \text{ and } r^v\in\mathfrak{R}^v \text{ with } r(v) \cap \sur{1}{\Theta}(\varphi) =r^v(v), \text{ for } v\in W^0 \setminus\{w_0\}\bigr\}.
%\end{multline*}
%
Observe that, as $v$ in $\mathfrak{C}^v$ is not $R_1$-connected, the $R_1$-components in $\mathfrak{Q}$ and $\mathfrak{C}^v$ are disjoint, and so the frame in the resulting stock belongs to $\mathcal{C}_1\otimes\mathcal{C}_2$. The resulting runs are $1$- and $2$-coherent and $1$-saturated on $W\cup W^v$. The runs are $2$-saturated on $W'\cup W^v$ if they are $2$-saturated on $W'\subseteq W$. % in the given stock.
\end{description}
%
%Note that since we do not graft anything onto the roots $v$ of the quasimodel candidate in the grafts themselves, the $R_1'$ and $R_1$ neighbours of $v$ coincide (and the $R_1$ relation in the initial quasimodel part of the tapered cactus is disconnected from the $R_1$ relation in the the grafts). 
%It then follows for both $i = 1,2$ that if runs in the $\mathfrak{C}^v$ are $i$-coherent, then they are $i$-coherent in $\mathfrak{C}$; if runs in the $\mathfrak{C}^v$ are $i$-saturated (except the root), then they are $i$-saturated everywhere in $\mathfrak{C}$ (except the root - to be clarified).
%
%We also assume that the types $r(v)$ and $r^v(v)$ \textcolor{red}{match}, for\dots

A \emph{fully grafted tapered $i$-cactus $((W,R_1, R_2), \mathfrak{R})$ for $\SIG$} results from grafting a compatible tapered $(3-i)$-cactus $\mathfrak{C}^{w_0}$ for $\SIG^{w_0} = \sur{i}{\Theta}(\SIG)$ onto the world $w_0$ of a tapered $i$-cactus $((W^0,R_1^0, R_2^0), w_0, \mathfrak{R}^0)$ for $\SIG$. Runs in a fully grafted tapered cactus are $i$-coherent and $i$-saturated for both $i = 1,2$ on the whole $W$.
\end{definition}

%we show the following claim, which implies Lemma~\ref{lemma:local}.
%
\begin{claim}\label{claim:tapered-cactus}
For all $\mathcal{L}$-formulas $\varphi$ and $i = 1,2$ with $\adp(\varphi) = \sur{i}{\adp}{(\varphi)}$,  each of \textup{\textbf{(L1)}}--\textup{\textbf{(L3)}} holds iff
\begin{description}
\item[(L4)] there is a \textup{(}fully grafted\textup{)} tapered $i$-cactus $((W,R_1, R_2), \mathfrak{R})$ for $\sub \varphi$ with $\varphi\notin r(w_0)$, for some $r\in\mathfrak{R}$. %\textup{;}
\end{description}
\end{claim}
\begin{proof}
We proceed by induction on alternation depth.
The basis, $\adp(\varphi) = 0$, is trivial as the formula belongs to $\mathcal{L}_1$ or $\mathcal{L}_2$ (or both, if it has no modalities). For the inductive step, consider an $\mathcal{L}$-formula~$\varphi$ and suppose that $\adp(\varphi) = k > 0$.
Assume that $\sur{1}{\adp}(\varphi) =  k$ (the other case is symmetrical) and that the claim has been shown for all formulas~$\psi$ with $\sur{2}{\adp}(\psi) < k$, which, by~\eqref{eq:adp}, includes all $\psi\in\sur{1}{\Theta}(\varphi)$. 

\smallskip

While $\textbf{(L2)}\Rightarrow\textbf{(L1)}$ follows from soundness of $\vdash_{L_1\otimes L_2}$, $\textbf{(L1)}\Rightarrow\textbf{(L3)}$ follows from closure under substitutions of $\vdash_{L_i}$ (cf.~$\textbf{(G1)}\Rightarrow\textbf{(G3)}$ in the proof of Lemma~\ref{lemma:global}).

\smallskip

$\textbf{(L3)}\Rightarrow\textbf{(L4)}$ By Kripke completeness of $L_1$, there is an $\textit{xd}$-model $\mathfrak{M}^0$ based on $(W^0,R^0_1)\in \mathcal{C}_1$ such that $\mathfrak{M}^0,w_0\models \Box_1^{\smash{\leq \sur{1}{\md}(\varphi)}} \sur{1}{(\hat{\QS}_1(\varphi))}$ and $\mathfrak{M}^0,w_0\not\models\sur{1}{\varphi}$ for some $w_0\in W^0$. 
Let $\SIG = \sub\varphi$.
For each $v\in W^0$, let $h(v) =\sur{1}{\md}(\SIG) - \dist_{\smash{R^0_1}}(w_0, v)$. 
Construct a set $\mathfrak{R}^0$  of tapered $\Phi$-runs through $((W^0, R^0_1),w_0)$: for each $e\in\bigcup_{v\in W^0}\Delta^0_v$, consider $r_e$ that maps each $v\in W^0$ with $e\in\Delta^0_v$ to the $\sur{1}{\Gamma}_{\smash{h(v)}}(\SIG)$-type $\type_{v,e}$ %of~$e$ at~$v$ 
with 
$\mathfrak{M},v\models \smash{\sur{1}{\type_{v,e}}[e]}$. Thus, $\mathfrak{Q} = ((W^0, R^0_1), w_0, \mathfrak{R}^0)$ is a tapered 1-quasimodel for $\SIG$ with $\varphi\notin r_{e}(w_0)$, for some $e\in\Delta_{w_0}^0$.

For each  $v\in W^0$ with $h(v)\geq 0$, take the $\Theta^1(\SIG)$-quasistate $\qs^v\in \QS_1(\varphi)$ with \mbox{$\mathfrak{M}^0,v\models \sur{1}{(\hat{\qs}^v)}$}.  
Denote $\SIG^v = \sur{1}{\Theta}(\SIG) \cap \sur{1}{\Gamma}_{\smash{h(v)}}(\SIG)$ and consider $\qs^v|_{\smash{\sur{1}{\Gamma}_{\smash{h(v)}}(\SIG)}}$, which is a $\SIG^v$-quasistate with $\qs^v|_{\smash{\sur{1}{\Gamma}_{h(v)}(\SIG)}} = \mathfrak{R}^0(v)|_{\smash{\sur{1}{\Theta}(\SIG)}}$. Let $\varphi^v$ be its realisability sentence. 
As $\sur{2}{\adp}(\varphi^v) < k$, 
we apply IH to 
$\varphi^v$ and obtain a tapered 2-cactus $\mathfrak{C}^v = ((W^v, R_1^v, R_2^v), v, \mathfrak{R}^v)$ for~$\SIG^v$
such that  $\qs^v|_{\smash{\sur{1}{\Gamma}_{h(v)}(\SIG)}} = \mathfrak{R}^v(v)$. 
For each $v\in W^0$ with $h(v)< 0$, take a tapered 2-cactus~$\mathfrak{C}^v$ for~$\emptyset$. 
So, for each $v\in W^0\setminus\{w_0\}$, we graft 
$\mathfrak{C}^v$ onto $v$ in 
$\mathfrak{Q}$ to obtain a tapered 1-cactus~$\mathfrak{C}$. A fully grafted tapered 1-cactus can be obtained by grafting,  in addition, $\mathfrak{C}^{w_0}$ onto $w_0$ in $\mathfrak{C}$.

\smallskip

$\textbf{(L4)}\Rightarrow\textbf{(L2)}$  Let $\mathfrak{C} = ((W,R_1, R_2), \mathfrak{R})$ be a fully grafted tapered $i$-cactus for $\sub\varphi$ with $\varphi\notin r(w_0)$, for some $r\in\mathfrak{R}$. We construct $\mathfrak{M}$ based on a $\mathcal{C}_1\otimes\mathcal{C}_2$-frame $(W, R_1, R_2)$ and domains~$\Delta_w = \mathfrak{R}_w$, for $w\in W$:  we set $\mathfrak{M},w \models P[r]$ iff $P(x)\in r(w)$, for each unary predicate letter~$P$, $w\in W$ and $r\in\mathfrak{R}_w$. 
By induction on the structure of formulas, we show that, for each $w\in W$, $r\in \mathfrak{R}_w$, $\chi\in \{ \psi, \neg\psi \mid \psi \in r(w)\}$, and $i = 1, 2$,
\begin{equation*}
\mathfrak{M},w\models \chi[r] \qquad\text{ iff }\qquad \sur{i}{\chi}\in \sur{i}{(r(w))} .
\end{equation*}
The basis of induction is by definition. The cases of the Boolean connectives and quantifiers are immediate from the definition of types and quasistates. The cases of the modalities follow from the fact that the runs in $\mathfrak{C}$ are $i$-coherent and $i$-saturated.
Thus, $\mathfrak{M},w_0\not\models \varphi$. 
\end{proof}

The claims on local consequence in Theorem~\ref{thm:transfer:equality-free} for \textit{cd}-models are obtained by a straightforward modification of the construction in Lemma~\ref{lemma:local} (the tapered runs are required to be total functions on the sets of worlds).
Finally, the transfer of the $d$-fmp follows from the proof of Claim~\ref{claim:tapered-cactus} (the fully grafted cactus is finite).

% !TEX root =  fusions-submission.tex

\section{Fusions of One-Variable Logics with Equality}
\label{sec:equality}
We show that once equality is allowed in one‑variable modal logics, the transfer results established in Section~\ref{sec:equality-free} no longer hold. Remarkably, the counterexamples we introduce next arise quite naturally. The propositional modal logic $\mathbf{Diff}$ of the `elsewhere'  operator was introduced by Von~Wright~\cite{vonWright1979Place} and is determined by the class $\mathcal{D}$ of all disjoint unions of frames $(W,R)$ with $R=\{(w,v) \in W\times W\mid w\not=v\}$; see also~\cite{DBLP:journals/jsyml/Rijke92}. We are interested here in the one-variable modal logics $\Logecd \mathcal{D}$ and~$\Logecd \mathcal{D}_{\text{fin}}$ 
with $\mathcal{D}_{\text{fin}}$ the class of disjoint unions of \emph{finite} frames in $\mathcal{D}$. Note that the propositional logic \textbf{Diff} has the fmp, and so is complete for both $\mathcal{D}$ and $\mathcal{D}_{\text{fin}}$. On the other hand, while one can easily prove the fmp for $\Logecd \mathcal{D}_{\text{fin}}$, the logic $\Logecd \mathcal{D}$ does not enjoy the fmp, and so  $\Logecd \mathcal{D}\subsetneq \Logecd \mathcal{D}_{\text{fin}}$~\cite{DBLP:conf/aiml/Hampson18}.	
It has been observed~\cite{DBLP:phd/ethos/Hampson16,DBLP:conf/aiml/Hampson18} that both logics can be reduced in polytime to the two-variable fragment of first-order logic with counting on arbitrary and finite domains, respectively, and so they are decidable in \textsc{coNExpTime}. 
\begin{theorem}\label{thm:nontransfer1}
	Both $\Logecd \mathcal{D}$ and $\Logecd \mathcal{D}_{\textup{fin}}$ are decidable but $\Logecd (\mathcal{D} \otimes \mathcal{D}_{\textup{fin}})$ is undecidable.
\end{theorem}

\begin{proof}
The proof is by reduction of the problem of solving Diophantine equations. For arbitrary polynomials $g,h$ with coefficients from $\mathbb{N}\setminus\{0\}$, it is undecidable whether $g=h$ has a solution in $\mathbb{N}\setminus\{0\}$~\cite{Davis1977}. 
Each polynomial equation can be rewritten into a set of elementary equations
$y=n$, $y=z_1+z_2$, and $y=z_1\times z_2$, with variables $y,z_1,z_2$ and $n\in \mathbb{N}\setminus\{0\}$. So we assume that a set $E=\{e_{1},\ldots,e_{n}\}$ 
of such equations is given.
	We construct a formula $\varphi_{E}$ which is satisfiable iff $E$ has a solution in~\mbox{$\mathbb{N}\setminus\{0\}$}.
	Let the modal operators $\Box$ and $\Box_{C}$ be interpreted over frames with accessibility relations $R$ in $\mathcal{D}$ and $R_{C}$~in~$\mathcal{D}_{\text{fin}}$, respectively. We set $\Box^{+}\varphi=\varphi \wedge \Box\varphi$, and likewise for $\Box_{C}$. Formula $\varphi_{E}$ will hardly need any alternation between modalities, the maximal alternation being of the form $\Box\Box_{C}\chi$. In fact, the only contribution of~$\mathcal{D}_{\text{fin}}$ is to 	supply the required \emph{finite} cardinality of the 
	connected components of its frames. 
	
	When discussing the construction of $\varphi_{E}$, we assume a world $w$ in the set $W$ of worlds reachable along~$R$. To $W$ attached are the models for $\Box_{C}$ with the $R_{C}$-connected component of~$w$ denoted by $W_{w}$. 
For a unary predicate letter $P$ and $w\in W$, we denote by $|P|_{w}$ the cardinality of the extension of~$P$ at $w$. In the sequel, for a closed formula $I$ and a unary predicate $P$, we use formula $\textsf{card}(I, P)$ defined as follows:
\begin{equation*}
	\Box^{+}\bigl[I \ \ \to \ \ \Box_{C}^{+}\forall x\,\bigl((x=c) \rightarrow \Box_{C}(x\ne c)\bigr) \ \ \land \ \ 
	\forall x\, \bigl(P(x) \leftrightarrow \Diamond_{C}^+(x=c)\bigr)\bigr],
\end{equation*}
where $c$ is a fresh constant. Formula $\textsf{card}(I, P)$ ensures that, at any world $w$ where $I$ holds, the number of worlds in the $R_C$-connected component of $w$ coincides with the number of domain elements in the extension of $P$ at~$w$, that is, $|W_{w}|=|P|_w$; see Fig.~\ref{fig:diophantus}~a), where the nodes representing the same domain element in different worlds are shown connected by the $R$- and $R_C$-relations.

For each variable $y$ used on the left- or right-hand side of equations  in $E$, we take unary predicate letters $P_{y}$ and $C_y$. Intuitively,  $|P_{y}|_{w}$ encodes the value of $y$, while $C_y$  is used to ensure that the value of $y$ is finite. To this end, $\varphi_{E}$ contains the conjunct 
\begin{equation}\label{eq:diophant:v1}
\Diamond^{+}(I_y \wedge \Box\neg I_y), \qquad\text{where } I_y = \exists x\, C_{y}(x),
\end{equation}
to mark a single world $w\in W$ with $I_y$.
We add the conjunct 
\begin{equation}
\textsf{card}(I_y, P_y)
\end{equation}
to $\varphi_{E}$ to ensure that $|P_y|_w = |W_w|$.
Note that $|P_{y}|_{w}$ is finite. The following conjuncts guarantee that each $P_y$ is non-empty  (only solutions in $\mathbb{N}\setminus\{0\}$ are admitted) and constant on~$W$:
\begin{equation}\label{eq:diophant:v3}
\exists x\,P_y(x) \quad\text{ and }\quad \Box^+\forall x\,\bigl(P_y(x) \rightarrow \Box P_y(x)\bigr).
\end{equation}
\begin{itemize}	
\item If $e\in E$ is of the form $y=n$, then we take $n$ constants $c_{1}^{e},\ldots,c_{n}^{e}$ interpreted differently
		and make the set of their interpretations the extension of $P_y$: 	
\begin{equation}\label{eq:diophant:e1-1}
		\forall x\, \bigl(P_{y}(x) \leftrightarrow  (x=c_1^{e})\lor \dots \lor (x=c_n^{e}) \bigr) \ \ \land \ \ \bigwedge\nolimits_{i \ne j}\hspace*{-0em}(c_{i}^{e}\ne c_{j}^{e}).
\end{equation}

\item If $e\in E$ is of the form $y=z_1+z_2$, we introduce unary predicate letterss $Q^e_1$ and $Q^e_2$ and add formulas 
\begin{equation}\label{eq:diophant:e2-1}
\textsf{card}(I_{z_i}, Q^e_i)\quad \text{ for } i = 1, 2,
\end{equation}
which ensure $|P_{z_i}|_{w_i} = |Q^{e}_i|_{w_i}$ at the world $w_i$ marked with~$I_{z_i}$. Next, we state that $Q^e_1$ and $Q^e_2$ are disjoint at all worlds $w\in W$ and empty at worlds $w\in W$ where $I_{z_i}$ does not hold:
\begin{equation}
\Box^+\forall x\,\bigl(Q^e_1(x) \to \Box^+\neg Q^e_2(x)\bigr),\qquad
\Box^+ \bigl(\exists x\,Q^e_1(x) \to I_{z_1} \bigr)\quad\text{ and }\quad
\Box^+ \bigl(\exists x\,Q^e_2(x) \to I_{z_2} \bigr).
\end{equation}
Finally, we add  the conjunct that ensures that $P_y$ is the union of $Q^e_1$ and $Q^e_2$: % at the world $w$ with $I^{e}$:
		\begin{equation}\label{eq:diophant:e2-3}
		\forall x\, \bigl(P_{y}(x) \leftrightarrow \Diamond^+ Q^e_1(x) \vee \Diamond^+ Q^e_2(x)\bigr),
		\end{equation}	
which implies that equation $e$ holds in the sense that $|P_{y}|_{w}=|P_{z_1}|_{w} + |P_{z_2}|_{w}$, for all worlds $w\in W$.	
\begin{figure}%
\centering%
\begin{tikzpicture}[>=latex,xscale=0.7, yscale=0.7,
	nd/.style={draw,circle,inner sep=0pt,minimum size=1.4mm,thick}]
\begin{scope}[xshift=0mm]
\begin{scope}[dashed,gray]
\draw(0,0) -- ++(0,-2);
\draw(1,0) -- ++(0,-2);
\draw(2,0) -- ++(0,-2);
\draw(3,0) -- ++(0,-2);
\end{scope}
\draw[rectangle,rounded corners=2mm,fill=gray!30,fill opacity=0.9] (-3.7,0) -- ++(-0.2,0.2) -- ++ (1.4,0) -- ++(0.2,-0.2);
\draw[dashed, rounded corners=2mm] (-3.7,0.4) rectangle +(1.4,-2.8); 
\draw[rectangle,rounded corners=2mm,fill=gray!30,fill opacity=0.9] (-3.7,0) -- ++(1.6,-1.6) -- ++ (1.4,0) -- ++(-1.6,1.6);
\node at (-3, -2.8) {\small $W$};
\node at (-1.4, -0.45) {\small $W_w$};
\draw[rectangle,rounded corners=2mm,fill=gray!60,fill opacity=0.9] (0.6,0) -- ++(-0.2,0.2) -- ++ (2.8,0) -- ++(0.2,-0.2);
\draw[rectangle,rounded corners=2mm,fill=gray!20,fill opacity=0.9,thick] (0.6,0.3) rectangle +(2.8,-0.6);
\draw[rectangle,rounded corners=2mm,fill=gray!60,fill opacity=0.9] (0.6,0) -- ++(1.6,-1.6) -- ++ (2.8,0) -- ++(-1.6,1.6);
\node at (-0.25,0.5) {$I$};
\begin{scope}[ultra thin,black]
\draw (-3,0) -- ++(1.6,0); \node at (4.5, 0) {$\Delta_{w}$};
\draw[rounded corners=1.5mm] (-1.4,-0.2) -- ++(0,0.4) -- ++(5.4,0) -- ++(0,-0.4) -- cycle;
\draw (-2.3,-0.7) -- ++(1.6,0);  \node at (5.2, -0.7) {$\Delta_{w'}$};
\draw[rounded corners=1.5mm] (-0.7,-0.9) -- ++(0,0.4) -- ++(5.4,0) -- ++(0,-0.4) -- cycle;
\draw (-1.6,-1.4) -- ++(1.6,0); \node at (5.9, -1.4) {$\Delta_{w''}$};
\draw[rounded corners=1.5mm] (0,-1.6) -- ++(0,0.4) -- ++(5.4,0) -- ++(0,-0.4) -- cycle;
\draw (-3,-2) -- ++(1.6,0); \node at (4.5, -2) {$\Delta_v$};
\draw[rounded corners=1.5mm] (-1.4,-2.2) -- ++(0,0.4) -- ++(5.4,0) -- ++(0,-0.4) -- cycle;
\end{scope}
\begin{scope}[draw=black]
\draw (0,0) -- ++(1.4, -1.4);
\draw (1,0) -- ++(1.4, -1.4);
\draw (2,0) -- ++(1.4, -1.4);
\draw (3,0) -- ++(1.4, -1.4);
\end{scope}
\node[nd,draw=black,fill=gray,label=left:{\footnotesize $w$\!}] (w0) at (-3,0) {};
\node[nd,draw=black,fill=gray,label={[label distance=-3pt]left:{\footnotesize\ $w'$}}] (w01) at (-2.3,-0.7) {};
\node[nd,draw=black,fill=gray,label={[label distance=-3pt]left:{\footnotesize\ $w''$}}] (w02) at (-1.6,-1.4) {};
%\node[nd,draw=black,fill=gray] (w1) at (6.5,-1) {};
%\node[nd,draw=black,fill=gray] (w11) at (5.5,-1.3) {};
\node[nd,draw=black,fill=gray,label=left:{\footnotesize $v$\!}] (w2) at (-3,-2) {};
\draw[thick,dashed] (w0) to node[left] {\scriptsize $R$} (w2);
\draw[thick] (w0) to node[above,pos=0.6,sloped,yshift=-0.9mm] {\scriptsize $R_C$} (w01);
\draw[thick] (w01) to node[above,pos=0.6,sloped,yshift=-0.9mm] {\scriptsize $R_C$} (w02);
%\draw[thick] (w1) to node[right] {\scriptsize $R$} (w2);
%\draw[thick] (w1) to node[below] {\scriptsize $R_C$} (w11);
%
\foreach \y in {0,-2} {
\foreach \x in {0,...,3}
\node[fill=black,circle,inner sep=0pt,minimum size=0.8mm] at (\x,\y) {};
}
\foreach \y in {-0.7,-1.4} {
\foreach \x in {0,...,3}
\node[fill=black,circle,inner sep=0pt,minimum size=0.8mm] at (\x-\y,\y) {};
}
\node[nd,fill=white,ultra thick] (c1) at (1,0) {};
\node[nd,fill=white,ultra thick] (c2) at (2.7,-0.7) {};
\node[nd,fill=white,ultra thick] (c3) at (4.4,-1.4) {};
\node (c) at (2.3,-2.8) {$c$};
\draw[thick,->] (c) to (c1);
\draw[thick,->] (c) to (c2);
\draw[thick,->] (c) to (c3);
\node (p) at (3.5,0.8) {$P$};
\draw[thick,->] (p) to (2,0.3);
\node at (0.5,-3.5) {a)};
\end{scope}
\begin{scope}[xshift=136mm]
\draw[dashed, rounded corners=2mm] (-6.7,0.4) rectangle +(1.5,-2.8); 
\node at (-5.9, -2.8) {\small $W$};
\node at (-4.5*0.8,-2.9) {$P_{z_2}$};
\node at (-2*0.8,-2.9) {$P_{z_1}$};
\node at (2.5*0.8,-2.9) {$P_{y}$};
\draw[rectangle,rounded corners=2mm,fill=gray!20] (-4*0.8+0.35,0.4) rectangle +(-1*0.8-0.7,-2.8);
\draw[rectangle,rounded corners=2mm,fill=gray!20] (-1*0.8+0.35,0.4) rectangle +(-2*0.8-0.7,-2.8);
\draw[rectangle,rounded corners=2mm,fill=gray!20] (5*0.8+0.35,0.4) rectangle +(-5*0.8-0.7,-2.8);
\draw[rectangle,rounded corners=1mm,fill=gray!50] (-0.2,-0.15) rectangle +(2*0.8+0.4,0.3);
\draw[rectangle,rounded corners=1mm,fill=gray!50] (3*0.8-0.2,-1.15) rectangle +(2*0.8+0.4,0.3);
\foreach \y in {0,...,-2} {
\foreach \x in {-5,...,5}
\node[fill=black,circle,inner sep=0pt,minimum size=0.8mm] at (\x*0.8,\y) {};
}
\begin{scope}[ultra thin,black]
\draw (-5.8,0) -- ++(1,0); \node at (5.2,0) {$\Delta_{w_1}$};
\draw[rounded corners=1.5mm] (-4.8,-0.2) -- ++(0,0.4) -- ++(9.4,0) -- ++(0,-0.4) -- cycle;
\draw (-5.8,-1) -- ++(1,0); \node at (5.2,-1) {$\Delta_{w_2}$};
\draw[rounded corners=1.5mm] (-4.8,-1.2) -- ++(0,0.4) -- ++(9.4,0) -- ++(0,-0.4) -- cycle;
\draw (-5.8,-2) -- ++(1,0); \node at (5.2,-2) {$\Delta_{w_3}$};
\draw[rounded corners=1.5mm] (-4.8,-2.2) -- ++(0,0.4) -- ++(9.4,0) -- ++(0,-0.4) -- cycle;
\end{scope}
\node[nd,fill=black,rectangle] (j1) at (-5*0.8,0) {};
\node[nd,fill=black,rectangle] (j2) at (-4*0.8,-1) {};
\node (j) at (-3.5, 0.9) {$J^e$};
\draw[thick,->] (j) to (j1);
\draw[thick,->] (j) to (j2);
\node (qp) at (1.5, 0.9) {$Q^e$};
\draw[->,thick] (qp) -- (0.5,0.15);
\draw[->,thick] (qp) -- (3.5,-0.85);
\begin{scope}\footnotesize
\node[nd,draw=black,fill=gray,label=left:{$w_1$}] (w0) at (-5.8,0) {};
\node[nd,draw=black,fill=gray,label=left:{$w_2$}] (w1) at (-5.8,-1) {};
\node[nd,draw=black,fill=gray,label=left:{$w_3$}] (w2) at (-5.8,-2) {};
\end{scope}
\draw[thick,dashed] (w0) to node[right] {\scriptsize $R$} (w1);
\draw[thick,dashed] (w1) to node[right] {\scriptsize $R$} (w2);
\node at (0.7,-3.5) {b)};
\end{scope}
\end{tikzpicture}%
\caption{a) Formula $\textsf{card}(I,P)$. %each horizontal line represents the domain at a world. 
b)~Equation $e$ of the form $y = z_1 \times z_2$ with $z_1 = 3$ and $z_2 = 2$.}\label{fig:diophantus}
\end{figure}	
		
\item	If $e\in E$ is of the form $y=z_1\times z_2$, then we introduce a unary predicate letter $C^{e}$ and a constant~$c^{e}$ and let $J^{e}(x) = C^{e}(x)\land (c^{e}=x)$.
	First, we ensure $|P_{z_2}|_{w}$ is equal to the number of worlds in $W$ satisfying $\exists x\, J^{e}(x)$: 
\begin{equation}\label{eq:diophant:e3-1}
	\forall x\, \bigl(P_{z_2}(x) \leftrightarrow \Diamond^{+}J^{e}(x)\bigr) \quad \text{ and } \quad \Box^{+}\forall x\, \bigl(J^e(x) \rightarrow \Box\neg J^e(x)\bigr),
\end{equation}
where, due to the second conjunct, $J^{e}(x)$ 
	cannot hold on the same domain element more than once in~$W$.
	Next, we introduce a unary predicate letter $Q^{e}$ and add formulas 
\begin{equation}
\textsf{card}(\exists x\, J^{e}(x), P_{z_1})\qquad \text{ and }\qquad \textsf{card}(
	\exists x\, J^{e}(x), Q^{e})
\end{equation}	
	 to
	ensure that $|P_{z_1}|_{w}=|Q^{e}|_{w}$ for all worlds $w\in W$ with $\exists x\, J^e(x)$. Then we state that the $Q^{e}$ are mutually disjoint at different worlds $w\in W$ and empty at worlds $w\in W$ where $\exists x\,J^e(x)$ does not hold:
	\begin{equation}
	\Box^{+}\forall x\,\bigl(Q^{e}(x) \rightarrow \Box \neg Q^{e}(x)\bigr)
	\quad\text{ and }
	\quad
	\Box^{+}\bigl(\exists x\, Q^{e}(x) \rightarrow \exists x\,J^e(x)\bigr).
	\end{equation}
	We have $|P_{z_1}|_{w}$ (encoded as $|Q^{e}|_{w}$) in exactly $|P_{z_2}|_{w}$-many worlds and take the sum by adding the conjunct
\begin{equation}\label{eq:diophant:e3-3}
	\forall x\, \bigl(P_{y}(x) \leftrightarrow \Diamond^{+}Q^{e}(x)\bigr),
\end{equation} 
	which implies that equation $e$ holds in the sense that $|P_{y}|_{w}=|P_{z_1}|_{w} \times |P_{z_2}|_{w}$, for all $w\in W$; see Fig.~\ref{fig:diophantus}~b), where only worlds in $W$ and their domains are shown, and $|W_{w_1}| = |W_{w_2}| = 2$.
\end{itemize}	

So, $\varphi_E$  contains~\eqref{eq:diophant:v1}--\eqref{eq:diophant:v3}, for each variable $y$ in equations in $E$, and one of~\eqref{eq:diophant:e1-1}, ~\eqref{eq:diophant:e2-1}--\eqref{eq:diophant:e2-3} or~\eqref{eq:diophant:e3-1}--\eqref{eq:diophant:e3-3},  for each equation $e\in E$. It should be clear that $\varphi_E$ is as required.
\end{proof}

Theorem~\ref{thm:nontransfer1} does not depend on the constant domain semantics since
$\mathcal{D}$ and $\mathcal{D}_{\text{fin}}$ do not feel the difference between expanding and constant domains. Note also that recursive axiomatisability is not preserved under fusions since we have shown that $\Logecd (\mathcal{D}\otimes \mathcal{D}_{\text{fin}})$ is not recursively enumerable.

For the global consequence, one can prove an even stronger non-transfer result.
\begin{theorem}
  Let $\mathscr{C}_1$ and $\mathscr{C}_2$ be classes of Kripke frames that are
  closed under disjoint unions and contain non-trivial frames. %and let
 % $L_i = \mathop{\mathit{Log}}_{\mathit{cd}}^{=} \mathscr{C}_i$, for
%  $i \in \{1,2\}$.  
Then the global consequence for $\Logecd (\mathcal{C}_1 \otimes \mathcal{C}_2)$ is
  undecidable.
\end{theorem}
The proof is given in the full version of this paper available at~\cite{arxiv}.
It is by reduction of Minsky machines and very similar to 
the known undecidability proof for temporal first-order logic with equality~\cite{Degtayrev2002}.

We conclude this section with two remarks. First, note that neither $\mathcal{D}$ nor $\mathcal{D}_{\textup{fin}}$ are propositionally modally definable classes of frames. In fact, as $\mathbf{Diff}$ is axiomatised by $p\rightarrow \Box\Diamond p$ and \mbox{$\Diamond p \rightarrow \Box(p \vee \Diamond p)$}, any frame 
$(W,E)$ with an equivalence relation $E$  validates $\textbf{Diff}$ (in other words, $\mathbf{Diff}\subseteq \mathbf{S5}$). The question arises whether our non‑preservation result continues to hold when one restricts the components to propositionally modally definable frame classes. We conjecture that it does, but at present we only have a proof for expanding domain semantics and under the assumption that polymodal logics are permitted as components. To demonstrate this,
let $\mathcal{C}_{\text{diff}}$ be the class of \emph{all} frames validating \textbf{Diff} and $\mathcal{C}_{\text{ftime}}$ the class of all frames validating the bimodal logic $\textbf{LTL}_{f}$ of all frames $(W,S,R)$ with $R$ a finite strict linear order and $S$ its successor relation. $\textbf{LTL}_{f}$ is axiomatised by taking the fusion of \textbf{GL.3} and $\textbf{Alt}_{1}$ and adding the standard induction axioms~\cite{goldblatt1987logics}. 
\begin{theorem}\label{thm:modaldef}
	$\Logexd \mathcal{C}_{\textup{diff}}$ and $\Logexd \mathcal{C}_{\textup{ftime}}$ are decidable but	 
	$\Logexd (\mathcal{C}_{\textup{diff}} \otimes \mathcal{C}_{\textup{diff}} \otimes \mathcal{C}_{\textup{ftime}})$
	is undecidable.
\end{theorem}
\begin{proof}
	Decidability of $\Logexd \mathcal{C}_{\text{diff}}$ is shown in~\cite{DBLP:conf/aiml/Hampson18}. For decidability (and Ackermann-hardness) of $\Logexd \mathcal{C}_{\text{ftime}}$, see~\cite{DBLP:journals/tocl/HampsonK15,DBLP:journals/apal/GabelaiaKWZ06}.
	Undecidability of $\Logexd (\mathcal{C}_{\text{diff}} \otimes \mathcal{C}_{\text{diff}} \otimes \mathcal{C}_{\text{ftime}})$ can be shown by adapting the reduction of Diophantine equations: first, observe that one can enforce that a frame 
	in $\mathcal{C}_{\text{diff}}$ is actually in~$\mathcal{D}$ using the formula
	\begin{equation*}
		\Box^{+}\forall x\,\bigl((x=c) \rightarrow \Box (x\ne c)\bigr),
	\end{equation*}
	which can only be satisfied if all worlds in $W$ are irreflexive. The fusion with
	$\mathcal{C}_{\text{ftime}}$ is used to enforce a \emph{finite} first-order domain for worlds in $W$. This can be achieved as follows: for another fresh constant~$c'$ and $\Box_{L}$ the modal operator interpreted by the strict order in $\mathcal{C}_{\text{ftime}}$, if the sentence 
	$\forall x\,\Diamond_{L}^{+}(x=c')$ is true at $w_{0}$, then the domain at $w_{0}$ cannot be larger than the length of the finite strict order seen from $w_{0}$. It is now straightforward to adapt the remaining steps of the undecidability proof. 
\end{proof}

We conjecture that using polymodal logics here is not essential. For instance, the standard reduction of polymodal logics to logics with one modality might be applicable~\cite{DBLP:journals/jsyml/KrachtW99}.

Second, the undecidability proof can be adapted to resolve an open problem concerning monodic fragments of first‑order modal logics~\cite{KuruczWolterZakharyaschevGabbay2003,DBLP:journals/corr/abs-2509-08165}. A key ingredient of the encoding of Diophantine equations is that we can enforce equal cardinality $|P|_w = |Q|_w$ of the extensions of predicates $P$ and~$Q$ at world~$w$. This can also be achieved directly using a bijective function in the extension of one-variable modal logic to, for instance, the monodic fragment of modal first-order logic over the two-variable fragment of $\mathbf{FO}$ with counting. Hence, it can be shown that this monodic fragment is undecidable over the class $\mathcal{D}_{\text{fin}}$. This shows that decidability does not always transfer from one-variable modal logic with equality to monodic fragments based on decidable fragments of first-order logic (recall that the two-variable fragment with counting is decidable~\cite{DBLP:journals/jolli/Pratt-Hartmann05}). Note that sufficient conditions for this kind of preservation from one-variable modal logics to monodic fragments are shown in~\cite{DBLP:journals/tocl/GhilardiNZ08}. 
\begin{theorem}
	$\Logecd \mathcal{D}_{\textup{fin}}$ is decidable, but the monodic fragment of the two-variable fragment with counting over $\mathcal{D}_{\textup{fin}}$ is undecidable.
\end{theorem}

% !TEX root =  fusions-submission.tex

\section{Fusions of Propositional Modal Logics Sharing an S5 Modality}

In this section, we generalise the fusion construction for one‑variable first‑order modal logics without equality to fusions of propositional modal logics sharing an \textbf{S5} modality. We begin by establishing a transfer theorem for logics admitting models that are homogeneous for the shared modality, and then apply it to (semi)commutators. In contrast to the transfer proof in Lemma~\ref{lemma:global}, in the global case we do not construct cactus models; instead, we merge the respective models of the component logics in a single step to obtain models of the fusion. This approach works provided we can ensure that the extensions of the relevant types in models of the component logics have matching cardinalities. A closely related technique was first used in~\cite{DBLP:conf/aiml/Wolter96} to obtain transfer results via algebraic semantics rather than Kripke semantics, thereby avoiding any reliance on the Kripke completeness of the component logics.

Consider a propositional modal language with  \textbf{S5} operator $\Box_{E}$ interpreted in Kripke models by the equivalence relation $E$. The $E$-equivalence class $[w]_{E}$ of $w\in W$ is defined as usual by setting $[w]_{E} = E(w)$. Let $\kappa$ be an infinite cardinal. 
Then $\mathfrak{M}$ is called \emph{$(E,\kappa)$-homogeneous} if, for every $E$-equivalence class $U$ in $\mathfrak{F}$ and every modal formula $\varphi$,  
$|\{ w\in U \mid \mathfrak{M},w\models \varphi\}|\in \{0,\kappa\}$.
\begin{definition}
	A propositional modal logic $L$ with \textbf{S5} operator $\Box_{E}$ admits $E$-homogeneous models  (for global consequence) if there exists an infinite cardinal $\kappa$ such that, for all $\kappa'\geq \kappa$, whenever $\vp \notin L$ (or, respectively, $\varphi\not\vdash^{\ast}_{L}\psi$), there exists an $(E,\kappa')$-homogeneous model $\mathfrak{M}$ based on an $L$-frame with $\mathfrak{M}\not\models\vp$ (or, respectively, with $\mathfrak{M}\models \varphi$ and $\mathfrak{M}\not\models\psi$).
\end{definition}
Note that modal logics that admit homogeneous models (for global consequence) are (globally) Kripke complete, by definition. To formulate our main results, assume for simplicity that $L_i$ is a bimodal logic in the propositional modal language $\mathcal{L}_{i,E}$ with modal operators $\Box_{i}$ and $\Box_{E}$, for $i=1,2$. Then $L_{1}\otimes_E L_{2}$
denotes the smallest modal logic containing $L_{1}\cup L_{2}$ in the language $\mathcal{L}_{1,2,E}$ with modal operators $\Box_{1},\Box_{2}$, and $\Box_{E}$.
\begin{theorem}\label{thm:globalprop}
	Let $L_{1}$ and $L_{2}$ admit $E$-homogeneous models \textup{(}for global consequence\textup{)}. Then 
	\begin{itemize}
		\item $L_{1} \otimes_E L_{2}$ is \textup{(}globally\textup{)} Kripke complete\textup{;}
		\item $L_{1} \otimes_E L_{2}$ is  \textup{(}globally\textup{)} decidable if $L_1$ and $L_2$ are both \textup{(}globally\textup{)} decidable.
	\end{itemize}
\end{theorem}
We give a detailed proof for the global consequence. The local case can be shown by a rather straightforward combination of the approaches to the local case in Section~\ref{sec:equality-free} and the global case in this section, see~\cite{arxiv} for details. Assume $L_{1}$ and $L_{2}$ admit $E$-homogeneous models for global consequence. 

Modulo replacing $\forall x$ with $\Box_{E}$, we use the same notation as in Section~\ref{sec:equality-free}.  
For instance, if $\FL$ is a finite set of $\mathcal{L}_{1,2,E}$-formulas closed under subformulas, then 
a \emph{$\FL$-type} is a maximal Boolean-consistent subset of $\{\psi,\neg\psi\mid \psi\in\FL\}$, 
and a \emph{$\FL$-quasistate} is a non-empty set of $\FL$-types such that
\begin{equation*}%\label{eq:exists:saturated}
%\tag{\textbf{qs}}
\Box_E\psi\in\type \quad\text{ iff }\quad \psi\in\type', \text{ for all } \type'\in\qs,\qquad\qquad \text{ for each }\Box_E\psi\in\FL \text{ and } \type\in\qs.
\end{equation*}
We set
\begin{equation*}
	\hat{\qs} \ \  = \ \   \Box_{E}\bigvee\nolimits_{\type\in\qs} \type \ \ \land \ \ \bigwedge\nolimits_{\type\in\qs} \Diamond_{E} t,
\end{equation*} 
where, for simplicity, we use $\type$ to denote the conjunction of all formulas in $\FL$-type $\type$. Again, $\hat{\QS}$ denotes $\bigvee \{ \hat{\qs} \mid \qs\in\QS \}$ for a set $\QS$ of $\FL$-quasistates. We also define surrogates in the obvious way by introducing for each formula of the form $\Box_i\psi$, a fresh proposition letter $p_{\Box_i\psi}$ called the \emph{surrogate} of $\Box_i\psi$.
To establish the claimed results in the global case, it suffices to show the following lemma.
\begin{lemma}\label{lem:eexpand}
Let $\varphi, \psi$ be $\mathcal{L}_{1,2,E}$-formulas. Denote $\FL = \sub\varphi\cup\sub\psi$. 
	Then the following are equivalent\textup{:}
	\begin{description}
		\item[(G$_E$1)] $\varphi\not\vdash_{L_1\otimes_{E} L_2}^*\psi$\textup{;}
		\item[(G$_E$2)] there is a model $\mathfrak{M}$ based on an $L_{1}\otimes_{E}L_{2}$-frame such that $\mathfrak{M}\models\varphi$ but $\mathfrak{M}\not\models\psi$\textup{;}
		\item[(G$_E$3)] there is a set $\QS$  of $\FL$-quasistates such that
		\begin{itemize}
			\item[\textup{\textbf{(G$_E$3.1)}}] $\sur{1}{\varphi}\land \sur{1}{\smash{\hat{\QS}}} \not\vdash_{L_1}^* \sur{1}{\psi}$,
			\item[\textup{\textbf{(G$_E$3.2)}}] $\sur{1}{\varphi}\land \sur{1}{\smash{\hat{\QS}}} \not\vdash_{L_1}^* \neg\sur{1}{\hat{\qs}_i}$, for each $\qs_i\in\QS$,
			\item[\textup{\textbf{(G$_E$3.3)}}] $\sur{2}{\smash{\hat{\QS}}} \not\vdash_{L_2}^* \neg\sur{2}{\hat{\qs}_i}$, for each $\qs_i\in\QS$.
		\end{itemize}
	\end{description}
\end{lemma}
\begin{proof}
	$\text{\bf(G$_E$2)}\Rightarrow\text{\bf(G$_E$1)}$ and $\text{\bf(G$_E$1)}\Rightarrow\text{\bf(G$_E$3)}$ are shown similar to the proof of %global one-variable case; see 
	Lemma~\ref{lemma:global}.

\smallskip	

\noindent $\text{\bf(G$_E$3)}\Rightarrow\text{\bf(G$_E$2)}$ Let $\QS$ be a set of  $\FL$-quasistates  such that {\bf(G$_E$3.1)}--\text{\bf(G$_E$3.3)} hold. 
	Choose an infinite $\kappa$ such that if $\chi\not\models_{L_{i}}^{\ast}\chi'$, then for every $\kappa'\geq \kappa$ there exists an $(E,\kappa')$-homogeneous model $\mathfrak{M}_{i}$ based on an $L_{i}$-frame with $\mathfrak{M}_{i}\models \chi$ and $\mathfrak{M}_{i}\not\models \chi'$, for $i=1,2$. By choosing $\kappa'\geq \kappa$
	sufficiently large and taking appropriate disjoint unions of models,
	it follows from {\bf(G$_E$3.1)}--{\bf(G$_E$3.3)} that, for each \mbox{$i=1,2$}, we can find an $(E_{i},\kappa')$-homogeneous model $\mathfrak{M}_{i}=(\mathfrak{F}_i,V_{i})$ with $\mathfrak{F}_i = (W_{i},R_{i},E_{i})$ such that $\mathfrak{F}_i\models L_{i}$ and the following conditions hold:
\begin{enumerate}
		\item $\mathfrak{M}_{1}\models \sur{1}{\varphi}\land \sur{1}{\smash{\hat{\QS}}}$, $\mathfrak{M}_{1}\not\models\sur{1}{\psi}$, and $|\{[w]_{E_{1}} \mid \mathfrak{M}_{1},w\models \sur{1}{\hat{\qs}}\}|=\kappa'$ for all $\qs\in\QS$;
		\item $\mathfrak{M}_{2}\models \sur{2}{\smash{\hat{\QS}}}$ and $|\{[w]_{E_{2}} \mid \mathfrak{M}_{2},w\models \sur{2}{\hat{\qs}}\}|=\kappa'$ for all $\qs\in\QS$.
\end{enumerate}
	It follows that we find a bijection $F\colon W_{1} \rightarrow W_{2}$ such that, for all $w\in W_{1}$: 
	\begin{enumerate}
		\item there is a (uniquely determined) $\FL$-type $\type$
		such that $\mathfrak{M}_{1},w\models \sur{1}{\type}$ and $\mathfrak{M}_{2},f(w)\models \sur{2}{\type}$;
		\item $[f(w)]_{E_{2}}= \{f(w') \mid w'\in [w]_{E_{1}}\}$. 
	\end{enumerate}
	Now construct a new model $\mathfrak{M}=(\mathfrak{F},V)$ by taking $\mathfrak{M}_{1}$, interpreting $\Box_{2}$ by $f^{-1}[R_{2}]= \{(w,w')\in W_{1} \mid (f(w),f(w'))\in R_{2}\}$, and setting $V(p)=f^{-1}(V_{2}(p))$ for all variables $p$ introduced as surrogates of subformulas of the form $\Box_2\chi$. Then, for $\mathfrak{F}=(W_{1},R_{1}, f^{-1}[R_{2}],E_{1})$, we have $\mathfrak{F}\models L_{1}\otimes_{E}L_{2}$ by definition, and one can show by induction for all formulas~$\chi\in\FL$ and $w\in W_{1}$:
	\begin{equation*}
		\mathfrak{M},w\models {\chi} \qquad\text{ iff }\qquad \mathfrak{M}_{1},w\models\sur{1}{\chi}\qquad\text{ iff }\qquad \mathfrak{M}_{2},f(w)\models\sur{2}{\chi}.
	\end{equation*}
	It follows that $\mathfrak{M}\models \varphi$ and $\mathfrak{M}\not\models\psi$, as required.
\end{proof}

The following lemma shows that the conditions of Theorem~\ref{thm:globalprop} are satisfied for (globally) Kripke complete (semi)commutators with \textbf{S5}.
\begin{lemma}\label{lem:thmappliesto} 
If a logic of the form $[L,\mathbf{S5}]^{\mathit{xd}}$ or $[L,\mathbf{S5}]^{\textit{cd}}$ is \textup{(}globally\textup{)} Kripke complete, then it admits $E$-homogeneous models \textup{(}for global consequence\textup{)}.
\end{lemma} 

\begin{proof}
	We give the proof for the local case, the proof in the global case is the same. Assume first that $[L , \textbf{S5}]^{\textit{xd}}$ is Kripke complete. Let $\kappa$ be any infinite cardinal such that, for every $\varphi\not\in [L , \textbf{S5}]^{\textit{xd}}$, there exists a model refuting $\varphi$ based on an $[L , \textbf{S5}]^{\textit{xd}}$-frame of size at most 
	$\kappa$. We show that $\kappa$ is as required. Assume $\kappa'\geq \kappa$ and $\varphi\not\in [L , \textbf{S5}]^{\textit{xd}}$. 
	By definition of $\kappa$, there exists a model $\mathfrak{M}=(\mathfrak{F},V)$ with $\mathfrak{F}=(W,R,E)$, $\mathfrak{M}\not\models\varphi$, $\mathfrak{F}\models [L , \textbf{S5}]^{\textit{xd}}$, and $W$ of size $\kappa$. 
	Then define $\mathfrak{M}_{\kappa'}=(\mathfrak{F}_{\kappa'},V_{\kappa'})$ with $\mathfrak{F}_{\kappa'}=(W',R',E')$ 
	by setting
	\begin{itemize}
		\item $W'= W \times \kappa'$;
		\item $(w,n)R' (w',n')$ if $wRw'$ and $n=n'$, for all $n,n'\in \kappa'$;
		\item $(w,n)E' (w',n')$ if $wEw'$, for all $n,n'\in \kappa'$;
		\item $V_{\kappa'}(p) = V(p) \times \kappa'$, for all proposition letters $p$.
	\end{itemize} 
	\emph{Claim.} $\mathfrak{F}_{\kappa'}\models [L , \textbf{S5}]^{\textit{xd}}$, $\mathfrak{M}_{\kappa'}\not\models \varphi$, and $\mathfrak{M}_{\kappa'}$ is an $(E,\kappa')$-homogeneous model. 
	
	\medskip
	\noindent 	Proof of Claim. By closure under disjoint unions, $(W',R')$ validates $L$. Also, $E'$ is easily seen to be an equivalence relation. We next show that $(W',R',E')$ satisfies $(\textit{lcom}_{R,E})$. Assume $(w,n)E'(w',m)$ and $(w',m)R'(w'',k)$. By definition, $wEw'$, $m=k$, and $w'R w''$. 
	Then, as $(W,R,E)$ satisfies $(\textit{lcom}_{R,E})$, we have a $v'$ with $wRv'$ and $v'Ew''$. But then $(w,n)R'(v',n)$ and $(v',n)E'(w'',k)$, as required.
	
	Let $\chi$ be any formula. Then one can easily show by induction that $\mathfrak{M},w\models \chi$ iff $\mathfrak{M}_{\kappa'},(w,n)\models \chi$ for all $n\in \kappa'$. It follows that $\mathfrak{M}_{\kappa'}\not\models \varphi$ and 
	$\mathfrak{M}_{\kappa'}$ is an $(E,\kappa')$-homogeneous model. This  finishes the proof of the claim and also the proof for logics of the form $[L , \textbf{S5}]^{\textit{xd}}$.
	
	\medskip
	
	Now assume a logic of the form $[L,\mathbf{S5}]^{\textit{cd}}$ is given and Kripke complete. The construction of $\mathfrak{M}_{\kappa'}$ is the same, using a model $\mathfrak{M}=(\mathfrak{F},V)$ with $\mathfrak{F}$ also satisfying $(\textit{rcom}_{R,E})$. So it suffices to show that $\mathfrak{M}_{\kappa'}$ then also satisfies $(\textit{rcom}_{R,E})$: assume $(w,n)R'(w',m)$ and $(w',m)E'(w'',k)$. By definition, $n=m$, $wRw'$, and $w'Ew''$.
	Then, as $(W,R,E)$ satisfies $(\mathit{rcom}_{R,E})$, we have a $v'$ with $wEv'$ and $v'Rw''$. But then $(w,n)E'(v',k)$ and $(v',k)R'(w'',k)$, as required.
\end{proof}

We obtain the following transfer result for (semi)commutators.
\begin{theorem}\label{thm:proptransfer}
Let $d \in\{\textit{xd}, \textit{cd}\}$ and $L_1$ and $L_2$ be consistent propositional modal logics such that $[L_i,\textbf{S5}]^d$ is \textup{(}globally\textup{)} Kripke complete, for both $i = 1,2$. Then
\begin{itemize}
\item $[L_{1},\textbf{S5}]^{d} \otimes_E [L_{2},\textbf{S5}]^{d}$ is \textup{(}globally\textup{)} Kripke complete\textup{;}
\item $[L_{1},\textbf{S5}]^{d} \otimes_E [L_{2},\textbf{S5}]^{d}$ is \textup{(}globally\textup{)} decidable iff $[L_{1},\textbf{S5}]^{d}$ and $[L_{2},\textbf{S5}]^{d}$ are \textup{(}globally\textup{)} decidable.
\end{itemize}
\end{theorem}

% !TEX root =  fusions-submission.tex

\section{Discussion}
We have established preservation and non‑preservation results for one‑variable first‑order modal logics and for related propositional modal logics. Many questions remain open:
\begin{itemize}
	\item The behaviour of fusions of full first‑order modal logic is still largely 
unexplored. A transfer theorem for canonicity is given in~\cite{shehtman2024fusions}, but, for example, the transfer of 
completeness is open.
    \item Monodic fragments extend one‑variable modal logics by allowing more expressive first‑order fragments while restricting modal operators to formulas with at most one free variable
\cite{KuruczWolterZakharyaschevGabbay2003,DBLP:journals/corr/abs-2509-08165}. 
We have shown that decidability does not transfer from the one‑variable 
fragment with equality to the monodic fragment over the two‑variable fragment 
of FO with counting. Positive preservation results are known~\cite{DBLP:journals/tocl/GhilardiNZ08},
and it is straightforward to adapt our proofs to, for instance, the monodic fragments based on the two-variable fragment of $\mathbf{FO}$ without equality or description logics such as \textit{ALC}, \textit{ALCI}, and \textit{SHI} with the universal role~\cite{DBLP:books/daglib/0041477,KuruczWolterZakharyaschevGabbay2003}. However, a systematic study of fusions of monodic fragments is still missing.
\item It would be valuable to identify sufficient conditions under which Kripke completeness and decidability transfer to fusions of one‑variable modal logics with equality. We note that many standard polymodal fusions with equality such 
as polymodal \textbf{K} and \textbf{S5} are in fact decidable~\cite{DBLP:journals/corr/abs-2509-08165}.
\item We have seen that, in many cases, completeness and decidability are preserved under fusions of propositional modal logics sharing an \textbf{S5} operator. Can this be extended to other modal logics?
\item Our analysis focused on Kripke complete modal logics. Since first‑order 
modal logics are often incomplete under standard Kripke semantics, it would be interesting to obtain preservation results that do not presuppose Kripke 
completeness. Techniques from~\cite{DBLP:conf/aiml/Wolter96,DBLP:journals/iandc/BaaderGT06} may be helpful here.
\end{itemize}

\bibliographystyle{eptcs}
\bibliography{fusions}

\begin{thebibliography}{10}
\providecommand{\bibitemdeclare}[2]{}
\providecommand{\surnamestart}{}
\providecommand{\surnameend}{}
\providecommand{\urlprefix}{Available at }
\providecommand{\url}[1]{\texttt{#1}}
\providecommand{\href}[2]{\texttt{#2}}
\providecommand{\urlalt}[2]{\href{#1}{#2}}
\providecommand{\doi}[1]{doi:\urlalt{https://doi.org/#1}{#1}}
\providecommand{\eprint}[1]{arXiv:\urlalt{https://arxiv.org/abs/#1}{#1}}
\providecommand{\bibinfo}[2]{#2}

\bibitemdeclare{article}{DBLP:journals/corr/abs-2509-08165}
\bibitem{DBLP:journals/corr/abs-2509-08165}
\bibinfo{author}{Alessandro \surnamestart Artale\surnameend},
  \bibinfo{author}{Christopher \surnamestart Hampson\surnameend},
  \bibinfo{author}{Roman \surnamestart Kontchakov\surnameend},
  \bibinfo{author}{Andrea \surnamestart Mazzullo\surnameend} \&
  \bibinfo{author}{Frank \surnamestart Wolter\surnameend}
  (\bibinfo{year}{2025}): \emph{\bibinfo{title}{Decidability in First-Order
  Modal Logic with Non-Rigid Constants and Definite Descriptions}}.
\newblock {\slshape \bibinfo{journal}{CoRR}} \bibinfo{volume}{abs/2509.08165},
  \doi{10.48550/arXiv.2509.08165}.
\newblock \eprint{2509.08165}.

\bibitemdeclare{article}{DBLP:journals/iandc/BaaderGT06}
\bibitem{DBLP:journals/iandc/BaaderGT06}
\bibinfo{author}{Franz \surnamestart Baader\surnameend},
  \bibinfo{author}{Silvio \surnamestart Ghilardi\surnameend} \&
  \bibinfo{author}{Cesare \surnamestart Tinelli\surnameend}
  (\bibinfo{year}{2006}): \emph{\bibinfo{title}{A new combination procedure for
  the word problem that generalizes fusion decidability results in modal
  logics}}.
\newblock {\slshape \bibinfo{journal}{Inf. Comput.}}
  \bibinfo{volume}{204}(\bibinfo{number}{10}), pp. \bibinfo{pages}{1413--1452},
  \doi{10.1016/J.IC.2005.05.009}.

\bibitemdeclare{book}{DBLP:books/daglib/0041477}
\bibitem{DBLP:books/daglib/0041477}
\bibinfo{author}{Franz \surnamestart Baader\surnameend}, \bibinfo{author}{Ian
  \surnamestart Horrocks\surnameend}, \bibinfo{author}{Carsten \surnamestart
  Lutz\surnameend} \& \bibinfo{author}{Ulrike \surnamestart Sattler\surnameend}
  (\bibinfo{year}{2017}): \emph{\bibinfo{title}{An Introduction to Description
  Logic}}.
\newblock \bibinfo{publisher}{Cambridge University Press},
  \doi{10.1017/9781139025355}.

\bibitemdeclare{article}{DBLP:journals/jair/BaaderLSW02}
\bibitem{DBLP:journals/jair/BaaderLSW02}
\bibinfo{author}{Franz \surnamestart Baader\surnameend},
  \bibinfo{author}{Carsten \surnamestart Lutz\surnameend},
  \bibinfo{author}{Holger \surnamestart Sturm\surnameend} \&
  \bibinfo{author}{Frank \surnamestart Wolter\surnameend}
  (\bibinfo{year}{2002}): \emph{\bibinfo{title}{Fusions of Description Logics
  and Abstract Description Systems}}.
\newblock {\slshape \bibinfo{journal}{J. Artif. Intell. Res.}}
  \bibinfo{volume}{16}, pp. \bibinfo{pages}{1--58}, \doi{10.1613/JAIR.919}.

\bibitemdeclare{article}{BezhanishviliKhan2026}
\bibitem{BezhanishviliKhan2026}
\bibinfo{author}{Guram \surnamestart Bezhanishvili\surnameend} \&
  \bibinfo{author}{Mashiath \surnamestart Khan\surnameend}
  (\bibinfo{year}{2024}): \emph{\bibinfo{title}{The Monadic Grzegorczyk
  Logic}}.
\newblock {\slshape \bibinfo{journal}{CoRR}} \bibinfo{volume}{abs/2412.10854},
  \doi{10.48550/arXiv.2412.10854}.
\newblock \eprint{2412.10854}.

\bibitemdeclare{inproceedings}{DBLP:conf/calco/DahlqvistP11}
\bibitem{DBLP:conf/calco/DahlqvistP11}
\bibinfo{author}{Fredrik \surnamestart Dahlqvist\surnameend} \&
  \bibinfo{author}{Dirk \surnamestart Pattinson\surnameend}
  (\bibinfo{year}{2011}): \emph{\bibinfo{title}{On the Fusion of Coalgebraic
  Logics}}.
\newblock In: {\slshape \bibinfo{booktitle}{Proc.\ of the 4th Int.\ Conf. on
  Algebra and Coalgebra in Computer Science ({CALCO} 2011)}}, {\slshape
  \bibinfo{series}{Lecture Notes in Computer Science}} \bibinfo{volume}{6859},
  \bibinfo{publisher}{Springer}, pp. \bibinfo{pages}{161--175},
  \doi{10.1007/978-3-642-22944-2\_12}.

\bibitemdeclare{incollection}{Davis1977}
\bibitem{Davis1977}
\bibinfo{author}{Martin \surnamestart Davis\surnameend} (\bibinfo{year}{1977}):
  \emph{\bibinfo{title}{Unsolvable Problems}}.
\newblock In \bibinfo{editor}{Jon \surnamestart Barwise\surnameend}, editor:
  {\slshape \bibinfo{booktitle}{Handbook of Mathematical Logic}},
  \bibinfo{publisher}{North-Holland}, \bibinfo{address}{Amsterdam}, pp.
  \bibinfo{pages}{567--594}, \doi{10.1016/S0049-237X(08)71115-7}.

\bibitemdeclare{article}{Degtayrev2002}
\bibitem{Degtayrev2002}
\bibinfo{author}{Anatoli \surnamestart Degtyarev\surnameend},
  \bibinfo{author}{Michael \surnamestart Fisher\surnameend} \&
  \bibinfo{author}{Alexei \surnamestart Lisitsa\surnameend}
  (\bibinfo{year}{2002}): \emph{\bibinfo{title}{Equality and Monodic
  First-Order Temporal Logic}}.
\newblock {\slshape \bibinfo{journal}{Stud Logica}}
  \bibinfo{volume}{72}(\bibinfo{number}{2}), pp. \bibinfo{pages}{147--156},
  \doi{10.1023/A:1021352309671}.

\bibitemdeclare{incollection}{FineSchurz1996}
\bibitem{FineSchurz1996}
\bibinfo{author}{Kit \surnamestart Fine\surnameend} \& \bibinfo{author}{Gerhard
  \surnamestart Schurz\surnameend} (\bibinfo{year}{1996}):
  \emph{\bibinfo{title}{Transfer Theorems for Multimodal Logics}}.
\newblock In \bibinfo{editor}{B.~Jack \surnamestart Copeland\surnameend},
  editor: {\slshape \bibinfo{booktitle}{Logic and Reality: Essays on the Legacy
  of Arthur Prior}}, \bibinfo{publisher}{Oxford University Press}, pp.
  \bibinfo{pages}{169--213}.

\bibitemdeclare{article}{Gabbay2003Fibred}
\bibitem{Gabbay2003Fibred}
\bibinfo{author}{Dov~M. \surnamestart Gabbay\surnameend}
  (\bibinfo{year}{1996}): \emph{\bibinfo{title}{Fibred Semantics and the
  Weaving of Logics, Part 1: Modal and Intuitionistic Logics}}.
\newblock {\slshape \bibinfo{journal}{J. Symb. Log.}}
  \bibinfo{volume}{61}(\bibinfo{number}{4}), pp. \bibinfo{pages}{1057--1120},
  \doi{10.2307/2275807}.

\bibitemdeclare{article}{DBLP:journals/igpl/GabbayS98}
\bibitem{DBLP:journals/igpl/GabbayS98}
\bibinfo{author}{Dov~M. \surnamestart Gabbay\surnameend} \&
  \bibinfo{author}{Valentin~B. \surnamestart Shehtman\surnameend}
  (\bibinfo{year}{1998}): \emph{\bibinfo{title}{Products of Modal Logics, Part
  1}}.
\newblock {\slshape \bibinfo{journal}{Log. J. {IGPL}}}
  \bibinfo{volume}{6}(\bibinfo{number}{1}), pp. \bibinfo{pages}{73--146},
  \doi{10.1093/JIGPAL/6.1.73}.

\bibitemdeclare{article}{DBLP:journals/apal/GabelaiaKWZ06}
\bibitem{DBLP:journals/apal/GabelaiaKWZ06}
\bibinfo{author}{David \surnamestart Gabelaia\surnameend}, \bibinfo{author}{Agi
  \surnamestart Kurucz\surnameend}, \bibinfo{author}{Frank \surnamestart
  Wolter\surnameend} \& \bibinfo{author}{Michael \surnamestart
  Zakharyaschev\surnameend} (\bibinfo{year}{2006}):
  \emph{\bibinfo{title}{Non-primitive recursive decidability of products of
  modal logics with expanding domains}}.
\newblock {\slshape \bibinfo{journal}{Ann. Pure Appl. Log.}}
  \bibinfo{volume}{142}(\bibinfo{number}{1-3}), pp. \bibinfo{pages}{245--268},
  \doi{10.1016/J.APAL.2006.01.001}.

\bibitemdeclare{article}{DBLP:journals/tocl/GhilardiNZ08}
\bibitem{DBLP:journals/tocl/GhilardiNZ08}
\bibinfo{author}{Silvio \surnamestart Ghilardi\surnameend},
  \bibinfo{author}{Enrica \surnamestart Nicolini\surnameend} \&
  \bibinfo{author}{Daniele \surnamestart Zucchelli\surnameend}
  (\bibinfo{year}{2008}): \emph{\bibinfo{title}{A comprehensive combination
  framework}}.
\newblock {\slshape \bibinfo{journal}{{ACM} Trans. Comput. Log.}}
  \bibinfo{volume}{9}(\bibinfo{number}{2}), pp. \bibinfo{pages}{8:1--8:54},
  \doi{10.1145/1342991.1342992}.

\bibitemdeclare{inproceedings}{DBLP:conf/lpar/GhilardiS03}
\bibitem{DBLP:conf/lpar/GhilardiS03}
\bibinfo{author}{Silvio \surnamestart Ghilardi\surnameend} \&
  \bibinfo{author}{Luigi \surnamestart Santocanale\surnameend}
  (\bibinfo{year}{2003}): \emph{\bibinfo{title}{Algebraic and Model Theoretic
  Techniques for Fusion Decidability in Modal Logics}}.
\newblock In: {\slshape \bibinfo{booktitle}{Proc.\ of the 10th Int.\ Conf.\ on
  Logic for Programming, Artificial Intelligence, and Reasoning (LPAR 2003)}},
  {\slshape \bibinfo{series}{Lecture Notes in Computer Science}}
  \bibinfo{volume}{2850}, \bibinfo{publisher}{Springer}, pp.
  \bibinfo{pages}{152--166}, \doi{10.1007/978-3-540-39813-4\_10}.

\bibitemdeclare{book}{goldblatt1987logics}
\bibitem{goldblatt1987logics}
\bibinfo{author}{Robert \surnamestart Goldblatt\surnameend}
  (\bibinfo{year}{1987}): \emph{\bibinfo{title}{Logics of time and
  computation}}.
\newblock \bibinfo{publisher}{Center for the Study of Language and
  Information}.

\bibitemdeclare{article}{DBLP:journals/logcom/GorankoP92}
\bibitem{DBLP:journals/logcom/GorankoP92}
\bibinfo{author}{Valentin \surnamestart Goranko\surnameend} \&
  \bibinfo{author}{Solomon \surnamestart Passy\surnameend}
  (\bibinfo{year}{1992}): \emph{\bibinfo{title}{Using the Universal Modality:
  Gains and Questions}}.
\newblock {\slshape \bibinfo{journal}{J. Log. Comput.}}
  \bibinfo{volume}{2}(\bibinfo{number}{1}), pp. \bibinfo{pages}{5--30},
  \doi{10.1093/LOGCOM/2.1.5}.

\bibitemdeclare{phdthesis}{DBLP:phd/ethos/Hampson16}
\bibitem{DBLP:phd/ethos/Hampson16}
\bibinfo{author}{Christopher \surnamestart Hampson\surnameend}
  (\bibinfo{year}{2016}): \emph{\bibinfo{title}{Two-dimensional modal logics
  with difference relations}}.
\newblock Ph.D. thesis, \bibinfo{school}{King's College London, {UK}}.
\newblock
  \urlprefix\url{https://kclpure.kcl.ac.uk/portal/en/studentTheses/two-dimensional-modal-logics-with-difference-relations/}.

\bibitemdeclare{inproceedings}{DBLP:conf/aiml/Hampson18}
\bibitem{DBLP:conf/aiml/Hampson18}
\bibinfo{author}{Christopher \surnamestart Hampson\surnameend}
  (\bibinfo{year}{2018}): \emph{\bibinfo{title}{The Bimodal Logic of Commuting
  Difference Operators Is Decidable}}.
\newblock In: {\slshape \bibinfo{booktitle}{Proc.\ of the 12th Conf.\ on
  Advances in Modal Logic (AiML 2018)}}, \bibinfo{publisher}{College
  Publications}, pp. \bibinfo{pages}{311--326}.
\newblock \urlprefix\url{http://www.aiml.net/volumes/volume12/Hampson.pdf}.

\bibitemdeclare{article}{DBLP:journals/tocl/HampsonK15}
\bibitem{DBLP:journals/tocl/HampsonK15}
\bibinfo{author}{Christopher \surnamestart Hampson\surnameend} \&
  \bibinfo{author}{Agi \surnamestart Kurucz\surnameend} (\bibinfo{year}{2015}):
  \emph{\bibinfo{title}{Undecidable Propositional Bimodal Logics and
  One-Variable First-Order Linear Temporal Logics with Counting}}.
\newblock {\slshape \bibinfo{journal}{{ACM} Trans. Comput. Log.}}
  \bibinfo{volume}{16}(\bibinfo{number}{3}), pp. \bibinfo{pages}{27:1--27:36},
  \doi{10.1145/2757285}.

\bibitemdeclare{article}{arxiv}
\bibitem{arxiv}
\bibinfo{author}{Roman \surnamestart Kontchakov\surnameend},
  \bibinfo{author}{Dmitry \surnamestart Shkatov\surnameend} \&
  \bibinfo{author}{Frank \surnamestart Wolter\surnameend}
  (\bibinfo{year}{2026}): \emph{\bibinfo{title}{Fusions of One-Variable
  First-Order Modal Logics (Extended Version)}}.
\newblock {\slshape \bibinfo{journal}{CoRR}} \bibinfo{volume}{abs/2603.04512},
  \doi{10.48550/arXiv.2603.04512}.
\newblock \eprint{2603.04512}.

\bibitemdeclare{article}{DBLP:journals/jsyml/KrachtW91}
\bibitem{DBLP:journals/jsyml/KrachtW91}
\bibinfo{author}{Marcus \surnamestart Kracht\surnameend} \&
  \bibinfo{author}{Frank \surnamestart Wolter\surnameend}
  (\bibinfo{year}{1991}): \emph{\bibinfo{title}{Properties of Independently
  Axiomatizable Bimodal Logics}}.
\newblock {\slshape \bibinfo{journal}{J. Symb. Log.}}
  \bibinfo{volume}{56}(\bibinfo{number}{4}), pp. \bibinfo{pages}{1469--1485},
  \doi{10.2307/2275487}.

\bibitemdeclare{article}{DBLP:journals/jsyml/KrachtW99}
\bibitem{DBLP:journals/jsyml/KrachtW99}
\bibinfo{author}{Marcus \surnamestart Kracht\surnameend} \&
  \bibinfo{author}{Frank \surnamestart Wolter\surnameend}
  (\bibinfo{year}{1999}): \emph{\bibinfo{title}{Normal Monomodal Logics Can
  Simulate All Others}}.
\newblock {\slshape \bibinfo{journal}{J. Symb. Log.}}
  \bibinfo{volume}{64}(\bibinfo{number}{1}), pp. \bibinfo{pages}{99--138},
  \doi{10.2307/2586754}.

\bibitemdeclare{incollection}{Kurucz2007Combining}
\bibitem{Kurucz2007Combining}
\bibinfo{author}{Agi \surnamestart Kurucz\surnameend} (\bibinfo{year}{2007}):
  \emph{\bibinfo{title}{Combining Modal Logics}}.
\newblock In \bibinfo{editor}{Patrick \surnamestart Blackburn\surnameend},
  \bibinfo{editor}{Johan \surnamestart van Benthem\surnameend} \&
  \bibinfo{editor}{Frank \surnamestart Wolter\surnameend}, editors: {\slshape
  \bibinfo{booktitle}{Handbook of Modal Logic}}, \bibinfo{publisher}{Elsevier},
  pp. \bibinfo{pages}{869--924}, \doi{10.1016/S1570-2464(07)80018-8}.

\bibitemdeclare{article}{DBLP:journals/lmcs/KuruczWZ25}
\bibitem{DBLP:journals/lmcs/KuruczWZ25}
\bibinfo{author}{Agi \surnamestart Kurucz\surnameend}, \bibinfo{author}{Frank
  \surnamestart Wolter\surnameend} \& \bibinfo{author}{Michael \surnamestart
  Zakharyaschev\surnameend} (\bibinfo{year}{2025}):
  \emph{\bibinfo{title}{Deciding the Existence of Interpolants and Definitions
  in First-Order Modal Logic}}.
\newblock {\slshape \bibinfo{journal}{Log. Methods Comput. Sci.}}
  \bibinfo{volume}{21}(\bibinfo{number}{4}), \doi{10.46298/LMCS-21(4:6)2025}.

\bibitemdeclare{book}{KuruczWolterZakharyaschevGabbay2003}
\bibitem{KuruczWolterZakharyaschevGabbay2003}
\bibinfo{author}{Agi \surnamestart Kurucz\surnameend}, \bibinfo{author}{Frank
  \surnamestart Wolter\surnameend}, \bibinfo{author}{Michael \surnamestart
  Zakharyaschev\surnameend} \& \bibinfo{author}{Dov~M. \surnamestart
  Gabbay\surnameend} (\bibinfo{year}{2003}):
  \emph{\bibinfo{title}{Many-Dimensional Modal Logics: Theory and
  Applications}}.
\newblock {\slshape \bibinfo{series}{Studies in Logic and the Foundations of
  Mathematics}} \bibinfo{volume}{148}, \bibinfo{publisher}{North Holland},
  \bibinfo{address}{Amsterdam}.

\bibitemdeclare{article}{DBLP:journals/jolli/Pratt-Hartmann05}
\bibitem{DBLP:journals/jolli/Pratt-Hartmann05}
\bibinfo{author}{Ian \surnamestart Pratt{-}Hartmann\surnameend}
  (\bibinfo{year}{2005}): \emph{\bibinfo{title}{Complexity of the Two-Variable
  Fragment with Counting Quantifiers}}.
\newblock {\slshape \bibinfo{journal}{J. Log. Lang. Inf.}}
  \bibinfo{volume}{14}(\bibinfo{number}{3}), pp. \bibinfo{pages}{369--395},
  \doi{10.1007/S10849-005-5791-1}.

\bibitemdeclare{article}{DBLP:journals/jsyml/Rijke92}
\bibitem{DBLP:journals/jsyml/Rijke92}
\bibinfo{author}{Maarten \surnamestart de~Rijke\surnameend}
  (\bibinfo{year}{1992}): \emph{\bibinfo{title}{The Modal Logic of
  Inequality}}.
\newblock {\slshape \bibinfo{journal}{J. Symb. Log.}}
  \bibinfo{volume}{57}(\bibinfo{number}{2}), pp. \bibinfo{pages}{566--584},
  \doi{10.2307/2275293}.

\bibitemdeclare{inproceedings}{ShehtmanShkatov2019}
\bibitem{ShehtmanShkatov2019}
\bibinfo{author}{Valentin \surnamestart Shehtman\surnameend} \&
  \bibinfo{author}{Dmitry \surnamestart Shkatov\surnameend}
  (\bibinfo{year}{2019}): \emph{\bibinfo{title}{On one-variable fragments of
  modal predicate logics}}.
\newblock In: {\slshape \bibinfo{booktitle}{Proc.\ of SYSMICS 2019}},
  \bibinfo{publisher}{ILLC, University of Amsterdam}, pp.
  \bibinfo{pages}{129--132}.

\bibitemdeclare{article}{Shehtman2023}
\bibitem{Shehtman2023}
\bibinfo{author}{Valentin \surnamestart Shehtman\surnameend} \&
  \bibinfo{author}{Dmitry \surnamestart Shkatov\surnameend}
  (\bibinfo{year}{2023}): \emph{\bibinfo{title}{Semiproducts, Products, and
  Modal Predicate Logics: Some Examples}}.
\newblock {\slshape \bibinfo{journal}{Doklady. Mathematics}}
  \bibinfo{volume}{108}(\bibinfo{number}{2}), pp. \bibinfo{pages}{411--418},
  \doi{10.1134/S1064562423701296}.

\bibitemdeclare{article}{shehtman2024fusions}
\bibitem{shehtman2024fusions}
\bibinfo{author}{Valentin \surnamestart Shehtman\surnameend} \&
  \bibinfo{author}{Dmitry \surnamestart Shkatov\surnameend}
  (\bibinfo{year}{2024}): \emph{\bibinfo{title}{Fusions of Canonical Predicate
  Modal Logics Are Canonical}}.
\newblock {\slshape \bibinfo{journal}{Advances in Modal Logic 2024, Short
  Papers}}, pp. \bibinfo{pages}{51--56}.

\bibitemdeclare{article}{Thomason1980}
\bibitem{Thomason1980}
\bibinfo{author}{Stephen~K. \surnamestart Thomason\surnameend}
  (\bibinfo{year}{1980}): \emph{\bibinfo{title}{Independent Propositional Modal
  Logics}}.
\newblock {\slshape \bibinfo{journal}{Studia Logica}} \bibinfo{volume}{39}, pp.
  \bibinfo{pages}{143--144}, \doi{10.1007/BF00370317}.

\bibitemdeclare{article}{Wajsberg33}
\bibitem{Wajsberg33}
\bibinfo{author}{M.~\surnamestart Wajsberg\surnameend} (\bibinfo{year}{1933}):
  \emph{\bibinfo{title}{{E}in erweiterter {K}lassenkalk\"{u}l}}.
\newblock {\slshape \bibinfo{journal}{Monatshefte f\"{u}r {M}athematik und
  {P}hysik}} \bibinfo{volume}{40}, pp. \bibinfo{pages}{113--126},
  \doi{10.1007/BF01708856}.

\bibitemdeclare{inproceedings}{DBLP:conf/aiml/Wolter96}
\bibitem{DBLP:conf/aiml/Wolter96}
\bibinfo{author}{Frank \surnamestart Wolter\surnameend} (\bibinfo{year}{1996}):
  \emph{\bibinfo{title}{Fusions of Modal Logics Revisited}}.
\newblock In: {\slshape \bibinfo{booktitle}{Proc.\ of the 1st Workshop on
  Advances in Modal Logic (AiML 1996)}}, \bibinfo{publisher}{{CSLI}
  Publications}, pp. \bibinfo{pages}{361--379}.

\bibitemdeclare{incollection}{vonWright1979Place}
\bibitem{vonWright1979Place}
\bibinfo{author}{Georg~Henrik \surnamestart von Wright\surnameend}
  (\bibinfo{year}{1979}): \emph{\bibinfo{title}{A Modal Logic of Place}}.
\newblock In \bibinfo{editor}{Ernest \surnamestart Sosa\surnameend}, editor:
  {\slshape \bibinfo{booktitle}{The Philosophy of Nicholas Rescher}},
  \bibinfo{publisher}{D. Reidel}, \bibinfo{address}{Dordrecht}, pp.
  \bibinfo{pages}{65--73}, \doi{10.1007/978-94-009-9407-2_9}.

\bibitemdeclare{article}{DBLP:journals/jsyml/ZanardoSS01}
\bibitem{DBLP:journals/jsyml/ZanardoSS01}
\bibinfo{author}{Alberto \surnamestart Zanardo\surnameend},
  \bibinfo{author}{Am{\'{\i}}lcar \surnamestart Sernadas\surnameend} \&
  \bibinfo{author}{Cristina \surnamestart Sernadas\surnameend}
  (\bibinfo{year}{2001}): \emph{\bibinfo{title}{Fibring: Completeness
  Preservation}}.
\newblock {\slshape \bibinfo{journal}{J. Symb. Log.}}
  \bibinfo{volume}{66}(\bibinfo{number}{1}), pp. \bibinfo{pages}{414--439},
  \doi{10.2307/2694931}.

\end{thebibliography}
\end{document}